\documentclass[a4paper,onecolumn,11pt,submitted=2017-05-09]{quantumarticle}
\pdfoutput=1
\usepackage[utf8]{inputenc}
\usepackage[english]{babel}
\usepackage[T1]{fontenc}
\usepackage{amsmath}
\usepackage{hyperref}

\usepackage{amsfonts,amsthm,array,bbm,bm,color,fancyvrb,graphicx,mathtools,nicematrix,pgfplots,placeins,tikz,xparse,xfrac}

\usetikzlibrary{quantikz,calc,arrows,shapes,shapes.geometric,positioning,calc,decorations,decorations.pathmorphing,decorations.text,decorations.markings,fadings,decorations.pathreplacing,matrix,tikzmark,fit}

\usepackage[numbers]{natbib}

\tikzset{%
  >=latex, 
  inner sep=0pt,%
  outer sep=2pt,%
  mark coordinate/.style={inner sep=1pt,outer sep=1pt,minimum size=3pt,
    fill=black}
}

\pgfplotsset{compat=1.18}
\hypersetup{colorlinks=true,linkcolor=blue,citecolor=blue,urlcolor=blue}
\NiceMatrixOptions{cell-space-limits = 2mm}

\newcommand{\comment}[1]{}

\newcommand{\order}[1]{\mathcal{O}\left({#1}\right)}

\newcommand{\fid}[1]{\operatorname{F}_\text{FLO}\left({#1}\right)}
\newcommand{\ext}[1]{\operatorname{\xi}_\text{FLO}\left({#1}\right)}
\newcommand{\extno}{\operatorname{\xi}_\text{FLO}}

\newcommand{\com}[2]{\left[#1,#2\right]} 
\newcommand{\acom}[2]{{\{#1, #2\}}} 

\newcommand{\abs}[1]{{\left\lvert{#1}\right\rvert}}
\newcommand{\norm}[1]{{\left\lVert{#1}\right\rVert}}

\newcommand{\re}[1]{{\operatorname{Re}\left({#1}\right)}}
\newcommand{\im}[1]{{\operatorname{Im}\left({#1}\right)}}

\newcommand{\ketbra}[2]{\ensuremath{\left | {#1}  \middle\rangle\!\middle\langle {#2}\right |}}
\DeclareDocumentCommand{\cyc}{ O{n} }{
  {\mathbb{Z}_n}
}
\DeclareDocumentCommand{\pnorm}{ m O{p} }{
  \norm{{#1}}_{#2}
}

\newcommand{\maj}[1]{{c_{#1}}}

\newcommand{\flohm}{\operatorname{\phi}}
\newcommand{\flohom}[1]{\flohm\left({#1}\right)}

\newcommand{\opI}{\operatorname{I}}

\DeclareMathOperator{\tr}{tr}

\newcommand{\changed}[1]{{#1}} 

\newtheorem{thm}{Theorem}
\newtheorem{lem}{Lemma}

\newtheorem{cor}{Corollary}

\theoremstyle{definition}
\newtheorem{defi}{Definition}
\newtheorem{updaterule}{Update rule}

\begin{document}

\date{30/04/2024}
\title{Improved simulation of quantum circuits dominated by free fermionic operations}

\author{Oliver Reardon-Smith}
\orcid{0000-0002-0124-1389}
\email{oliverreardonsmith@gmail.com}
\affiliation{Center for Theoretical Physics, Polish Academy of Sciences, Al. Lotnikow 32/46, 02-668 Warszawa, Poland}

\author{Micha{\l} Oszmaniec}
\affiliation{Center for Theoretical Physics, Polish Academy of Sciences, Al. Lotnikow 32/46, 02-668 Warszawa, Poland.}
\affiliation{NASK National Research Institute, ul. Kolska 12
01-045 Warszawa, Poland}
\orcid{0000-0002-4946-6835}

\author{Kamil Korzekwa}
\orcid{0000-0002-0683-5469}
\affiliation{Faculty of Physics, Astronomy and Applied Computer Science, Jagiellonian University, 30-348 Kraków, Poland}

\begin{abstract}

  We present a classical algorithm for simulating universal quantum circuits composed of ``free'' nearest-neighbour matchgates or equivalently fermionic-linear-optical (FLO) gates, and ``resourceful'' non-Gaussian gates. We achieve the promotion of the efficiently simulable FLO subtheory to universal quantum computation by gadgetizing controlled phase gates with arbitrary phases employing non-Gaussian resource states. Our key contribution is the development of a novel phase-sensitive algorithm for simulating FLO circuits. This allows us to decompose the resource states arising from gadgetization into free states at the level of statevectors rather than density matrices. The runtime cost of our algorithm for estimating the Born-rule probability of a given quantum circuit scales polynomially in all circuit parameters, except for a linear dependence on the newly introduced \emph{FLO extent}, which scales exponentially with the number of controlled-phase gates. More precisely, as a result of the decompostions we find for relevant resource states, the runtime doubles for every maximally resourceful (e.g., swap or CZ) gate added. Crucially, this cost compares very favourably with the best known prior algorithm, where each swap gate increases the simulation cost by a factor of approximately 9. For a quantum circuit containing arbitrary FLO unitaries and $k$ controlled-Z gates, we obtain an exponential improvement $\order{4.5^k}$ over the prior state-of-the-art.
\end{abstract}

\maketitle


\section{Introduction}

As our capabilities in the control and manipulation of quantum systems extend, so do the demands we make of the simulation methods we use to understand them. Despite the interest in the potential use of quantum simulators to simulate quantum systems~\cite{lloyd1996universal,altman2021quantum}, for now the availability, cost and reliability of classical computers make them by far the predominant tool for the simulation of quantum mechanics. Although the most obvious methods to classically simulate quantum computations have runtime and memory requirements that scale linearly with the dimension of the Hilbert space in question (i.e., exponentially in the number of qubits involved), this scaling is \emph{not} necessary. It is obvious that quantum circuits containing no entangling gates may be classically simulated in polynomial time, while other seminal examples have been known since the turn of the millennium, including the celebrated examples of Clifford/stabilizer circuits~\cite{gottesman-knill-1999,aaronson-gottesman-2004} and matchgate/fermionic linear optical (FLO) circuits~\cite{valiant-2001,terhal2002Classical}. 

Although these examples play a central role in randomized benchmarking \cite{Dekert2009,Helsen2022matchgate}, error correction \cite{ErrCorrTutorial2019}, and protocols such as classical shadows \cite{huang20} and their fermionic generalization \cite{zhao21}, the very restrictions which make them efficiently classically simulable cause them to fail to be computationally universal. Therefore, they form ``gilded cages'' within which we can simulate anything we wish to, but outside which lie most of the things we wish to simulate. This fact has been used to propose schemes for quantum advantage based on sampling for random Clifford~\cite{YJS2019} or free-fermion circuits~\cite{PRXQuantum.3.020328} initialized by product of suitable  magic states, promoting the respective sub-theories to computational universality. It is then natural that significant attention has been paid to generalizing these simulation methods to as broad a class of circuits as possible. 

On one hand a number of simulation methods covering restricted families of circuits extended by noisy variants of extra gates or states, have been developed for both Clifford \cite{Franca21} and FLO circuits \cite{deMelo2013,Oszmaniec2014}. On the other hand, in works such as~\cite{PRXQuantum.2.010345,PhysRevLett.116.250501,bravyi-2019,9605307,PhysRevX.6.021043,PRXQuantum.3.020361,Mitarai_2021}, a class of simulation algorithms was developed for universal quantum circuits with runtime that scales polynomially in the relevant parameters, except for an exponential dependence on the amount of \emph{quantum resources} such as magic or entanglement present in a circuit. In this way, one can interpolate smoothly between ``best-case'' efficiently simulable examples and the ``worst-case'' situation, where runtime is exponential in the number of qubits, enabling the exploration of an interesting intermediate regime.

These prior works have explored such ideas in depth in the cases of Clifford/stabilizer ``magic'' and entanglement, but less attention has been paid to classical simulation based on augmenting FLO circuits with their own ``magic''. The primary example that we are aware of has a somewhat poor runtime scaling, multiplying the required time by a factor of $9$ for each controlled-Z gate added to the circuit~\cite{PhysRevResearch.4.043100}. In this paper, we partially remedy this situation by adapting methods developed for efficient simulation of Clifford+T circuits~\cite{bravyi-2019,PRXQuantum.3.020361} to the matchgate+magic scenario. While the algorithms of Ref.~\cite{PhysRevResearch.4.043100} work at the level of density operators, ours work directly at the level of statevectors, leading to a significant improvement in performance. More precisely, we obtain a runtime which doubles for each controlled-Z gate in the circuit, while maintaining a polynomial dependence on the number of qubits and the number of ``free'' fermionic linear optical gates in the circuit. 

This dramatic improvement in runtime comes at the cost of requiring  more complicated simulation subroutines for the fermionic linear optics, since we require phase-sensitive FLO simulation, and the standard methods, known since the early 2000s~\cite{bravyi-kitaev-2000,terhal2002Classical} are \emph{not} phase-sensitive. We note that the authors of Ref.~\cite{bravyi-gosset-17-impurity} developed phase-sensitive FLO simulation methods, however their approach was rather different to ours, being targeted at problems related to quantum impurity models. In particular, their phase-sensitivity is a result of keeping an explicit phase-reference state. In the context of impurity models, where one expects the state at all times to be not too far from the ground state of the initial Hamiltonian, this unperturbed ground state serves as a natural phase-reference. In our context, with a focus on quantum circuits, this approach would be more troublesome, since in a quantum circuit it is natural to allow the state to evolve to a state orthogonal to any fixed reference. Indeed, it is usually not possible to efficiently decide whether a particular quantum circuit will evolve a state to one orthogonal to some reference state. In practical non-idealised computing, the problem is even worse, as significant numerical errors can arise when the state evolves into one close to, but not exactly orthogonal to a fixed phase reference. To address these issues, we have developed novel phase-sensitive FLO simulation routines, more suitable for application to quantum circuits, inspired by the ``CH-form'' developed for phase-sensitive Clifford/stabilizer simulation~\cite{bravyi-2019}. We anticipate that these phase-sensitive FLO simulation subroutines may also be applied more broadly. 

Classical simulation of quantum circuits can be understood in a variety of ways, but our main contribution is an algorithm for additive-error Born-rule probability estimation. Although perhaps the problem of \emph{sampling} from the Born-rule probability distribution is more obvious, probability estimation has natural applications in verification and validation of quantum computers, e.g., by cross-entropy benchmarking~\cite{Boixo-xeb-2018}, and has been employed to estimate the energies of Hamiltonians which may be written as sums of polynomially-many binary observables, e.g., in the Pauli-operator basis~\cite{PRXQuantum.3.020361}. The runtime of our algorithm is polynomial in all relevant parameters (the error, failure probability, qubit count, gate count, etc.) other than a factor exponential in the number of non-FLO gates in the circuit, which appears in the form of a \emph{FLO-extent}. We define this quantity in analogy to the stabilizer extent known from Ref.~\cite{bravyi-2019} and related works. The FLO-extent may be viewed as an upper bound on the closely related (approximate)-FLO rank, although we do not focus on this connection here.

Our results are formulated for quantum circuits, directly extending the class of quantum computations that can be simulated using current classical computers. However, extensions beyond the FLO subtheory also naturally represents the evolution of fermions evolving under  non-Gaussian transformations and, in the specific case of particle preserving FLO gates, the introduction of non-Gaussian gates can induce interactions between fermions~\cite{oz2017}. Therefore, we expect that our results will be applicable to classical simulations of weakly interacting fermions. Indeed the operations we consider are directly relevant for quantum simulations of chemistry and many-body systems~\cite{doi:10.1126/science.abb9811,PhysRevLett.120.110501,PhysRevApplied.9.044036,Dallaire-Demers_2019,Phasecraft2022,UnbiasingMonteCarlo2022,FollowUpMatchShadowGoogle2022} 

The paper is organized as follows. First, in Sec.~\ref{sec:preliminaries}, we present the background material and known results on fermionic linear optical circuits, the corresponding fermionic Gaussian states, and the matchgate magic states. We also introduce the notion of FLO extent and describe some of its properties. Then, in Sec.~\ref{sec:phase-sensitive-flo-sim}, we present our phase-sensitive simulator for FLO circuits. Our main result, the algorithm for estimating Born-rule probabilities of universal quantum circuits composed of FLO unitaries and controlled-phase gates, is described in Sec.~\ref{sec:simulating-universal-circuits}. We discuss its performance in Sec.~\ref{sec:discussion}, while Sec.~\ref{sec:outlook} contains outlook for future research. Most of the proofs of technical lemmas can be found in Appendices~\ref{app:lie}-\ref{lem:app-bounding-norm-alpha-y}. An implementation of the phase-sensitive fermionic linear optics simulation routines may be found at Ref.~\cite{ors-flo-github}.


\section{Preliminaries}
\label{sec:preliminaries}


\subsection{Fermionic linear optics}

Our work can be understood from the two points of view: we either investigate a quantum circuit consisting of matchgates acting on nearest neighbour qubits (in a 1D architecture) and controlled-phase gates applied to a system of $n$-qubits; or, equivalently, we study a system consisting of $n$ fermionic modes with fermionic linear optical (FLO) operations augmented by the addition of some non-fermionic linear operations. These viewpoints are linked by the Jordan-Wigner transformation \cite{jordan-wigner-1928}, which maps operators on the fermionic system to those of the qubit system and vice-versa.

On the fermionic side, we will adopt the notation of Refs.~\cite{terhal2002Classical,bravyi-kitaev-2000}. A system of $n$ fermionic modes, labelled with indices from $0$ to $n-1$, may be described by $n$ annihilation and $n$ creation operators satisfying the canonical anti-commutation relations:
\begin{align}
  \acom{a_j}{a_k} &= 0, & \acom{a_j^\dagger}{a_k^\dagger} &= 0, &  \acom{a_j}{a_k^\dagger} &= \delta_{jk} I.\label{eqn:creation-annihilation-anticomm-relations}
\end{align}
Using the above, it is then useful to define $2n$ Majorana fermion operators~\cite{bravyi-kitaev-2000}, which satisfy their own anti-commutation relations: 
\begin{align}
   c_{2j} &= (a_j + a_j^\dagger), & c_{2j+1} &= i(a_j^\dagger - a_j), & \acom{c_j}{c_k} &= 2\delta_{jk} I. \label{eqn:majorana-anticomm-relations}
\end{align}
The Majorana fermion operators are easily seen to be both Hermitian and unitary, square to the identity, and anti-commute with each other. These properties are strongly reminiscent of the Pauli operators, defined for a system qubits, and as we shall see, this similarity is not a coincidence. For a system of $n$ qubits, a Pauli string is an integer power of $i$ multiplied by a tensor product of single qubit Pauli matrices, which consist of
\begin{align}
  I &= \begin{pmatrix}1 & 0 \\0 & 1\end{pmatrix}, &\changed{\sigma_x} &= \begin{pmatrix}0 & 1 \\1 & 0\end{pmatrix}, &\changed{\sigma_y} &= \begin{pmatrix}0 & -i \\i & 0\end{pmatrix}, &\changed{\sigma_z} &= \begin{pmatrix}1 & 0 \\0 & -1\end{pmatrix}.
\end{align}
\changed{We will occasionally employ the standard notations $X_j$, $Y_j$ and $Z_j$ to indicate not an individual Pauli matrix, but a $2^n\times 2^n$ matrix formed of a single non-trivial Pauli matrix acting on the $j^\text{th}$ qubit and a tensor product of $n-1$ identity operators on the other qubits.} Examples of Pauli strings include the operator which is an identity on each qubit, the three additional non-trivial single qubit Paulis for each qubit, and products such as $i^2 X_1 Z_3 = -X \otimes I \otimes Z\otimes\ldots$, where the ellipsis indicates an identity operator on each of the remaining qubits. The set of $n$-qubit Pauli strings forms a group with the group operation being given by the standard matrix multiplication. This group plays a central role in quantum information, with Abelian subgroups being the stabilizers of the so-called stabilizer states, while the unitary operators mapping Pauli strings to Pauli strings form the Clifford group, which is ubiquitous in quantum information and error correction. Of particular significance is the fact that the Pauli strings form a basis for the space of $n$ qubit operators as a complex vector space, while the subset of Hermitian Pauli strings form a basis for the $n$ qubit Hermitian operators as a real vector space.

Since our focus is on simulating quantum circuits, we will view the $n$ qubit system as primary, and consider the Jordan-Wigner transformation as a relabelling, associating fermionic operators to real underlying qubit operators. Of course this viewpoint is a matter of notational convenience, and our work applies equally well to fundamentally fermionic systems. With this convention, we obtain the following expressions for the Majorana fermion operators in terms of $n$ qubit Pauli strings:
\begin{align}
  c_{2j} &= \left(\prod_{k=0}^{j-1} Z_k\right)X_j, & c_{2j+1} &= \left(\prod_{k=0}^{j-1} Z_j\right)Y_j,
\end{align}
where $\prod_{k=0}^{-1} Z_k = \opI$. It is easily verified that these $2n$ operators satisfy the anti-commutation relations from Eq.~\eqref{eqn:majorana-anticomm-relations}, required of Majorana fermion operators. 

A self-adjoint Hamiltonian $H$ on fermionic system with $n$ modes is called fermionic linear optical if it can be written in the form 
\begin{align}
  H &= \frac{i}{4}\sum_{j,k=0}^{2n-1} \alpha_{jk} c_j c_k \label{eqn:flo-hamiltonian-definition}
\end{align}
where the factor of $1/4$ is conventional and $\alpha$ is a real and anti-symmetric $2n\times 2n$ matrix \footnote{Note that due to the anti-commutation relations between Majorana operators, Eq.~\eqref{eqn:flo-hamiltonian-definition} constitutes the most general form of a Hamiltonian that is quadratic in Majorana operators (or equivalently in creation and annihilation operators), up to physically irrelevant term proportional to identity operator $I$.}.
It can be checked that for Hamiltonians $H,H'$ as in Eq. \eqref{eqn:flo-hamiltonian-definition}, their commutator $[H,H']$ also belongs to this class (up to multiplication by factors of $i$).  We therefore define the Lie algebra 
\begin{align}
    \mathfrak{flo}(n) &= \left\{\frac{1}{4} \sum_{j,k=0}^{2n-1} \alpha_{jk}c_j c_k\middle|\, \alpha\in\mathbb{R}^{2n\times 2n},\,\alpha^T + \alpha = 0\right\},
\end{align}
and note that the set of FLO Hamiltonians is just $i\mathfrak{flo}(n)$. In Appendix~\ref{app:lie}, we review the properties of this Lie algebra and the corresponding Lie group $\operatorname{FLO}(n)$ that it generates.

A unitary operator $U$ is called fermionic linear optical (or fermionic Gaussian), $U\in\operatorname{FLO}(n)$, if it is generated by a FLO Hamiltonian,
\begin{align}
    \label{eq:FLO_unitary}
  U = \exp\left(i H\right) = \exp\left(-\frac{1}{4} \sum_{j,k=0}^{2n-1} \alpha_{jk} c_j c_k\right),
\end{align}
where $\alpha$ will be referred to as the \emph{generating matrix} of $U$. An important fact is that conjugation by FLO unitaries maps Majorana operators to sums of Majorana operators~\cite{bravyi-kitaev-2000,terhal2002Classical}, e.g.,
\begin{align}
  U c_j U^\dagger &= \sum_{k=0}^{2n-1} R_{jk} c_k, \label{eqn:notation-FLO-conjugation}
\end{align}
for some matrix $R$. In fact, the matrix $R$ is a real and special orthogonal matrix. Here, we note a minor disagreement of notation in the literature: Refs.~\cite{terhal2002Classical,bravyi-kitaev-2000,josaMatchgate2008,bravyi-2005-lagrangian-rep} define FLO unitaries in the same we do (as those unitaries generated by Hamiltonians quadratic in the Majorana fermion operators), however other papers, such as Ref.~\cite{bravyi-gosset-17-impurity}, add the individual Majorana fermion operators $c_j$ to the group of FLO unitaries. These operators are \emph{odd}, and have the effect of adding $1$ ($\operatorname{mod} 2$) to the number of fermions in the relevant fermionic mode. The effect at the level of Eq.~\eqref{eqn:notation-FLO-conjugation} is to add \emph{reflections} to what was previously the special orthogonal group, in this way the full orthogonal group is generated. 

The map $\flohm$ which takes a FLO unitary $U$ and outputs the special orthogonal matrix satisfying
\begin{align}
    \label{eq:phi}
  U c_j U^\dagger &= \sum_{k=0}^{2n-1} \flohom{U}_{jk} c_k,
\end{align}
is a Lie group \changed{anti-}homomorphism. \changed{Our indexing in equations~\eqref{eqn:notation-FLO-conjugation} and~\eqref{eq:phi} is chosen to match the existing literature, however if instead $\phi(U)$ is transposed in~\eqref{eq:phi} it becomes a true homomorphism}. In Appendix~\ref{app:lie}, we show that the kernel of this homomorphism consists of the two unitaries $\pm I$. The \changed{anti-}homomorphism $\flohm$ may be made more explicit by writing the FLO unitary in terms of its generating Hamiltonian, then
\begin{align}
  \flohom{\exp\left(-\frac{1}{4} \sum_{j,k=0}^{2n-1} \alpha_{jk} c_j c_k\right)} = \exp\left(\alpha\right).
\end{align}
In particular, note that for any real special orthogonal matrix $R$ (of appropriate dimension), there exists some anti-symmetric $\alpha$ such that $\exp(\alpha) = R$, and so a FLO unitary exists that applies the rotation $R$ to the Majorana fermion operators.

Since under the Jordan-Wigner transformation the group of products of Majorana operators is exactly the group of Pauli strings, it is interesting to note that the Clifford group may be equivalently defined as the unitary operators mapping products of Majorana operators to products of Majorana operators under conjugation, while, due to equation~\eqref{eqn:notation-FLO-conjugation}, $\operatorname{FLO}(n)$ is eactly the group of unitaries mapping \emph{sums} of Majorana operators to sums of Majorana operators. Informally then, the Clifford group is to multiplication of Majorana operators as the FLO group is to addition. An important fact is that under the correspondence given by the Jordan-Wigner transformation, the group of FLO unitaries we have defined is exactly the group generated by nearest neighbour \emph{matchgates}~\cite{valiant-2001}, i.e., unitaries  of the form 
\begin{align}
    \label{eq:matchgate}
    G(A,B) &= \begin{pmatrix}
        A_{11} & 0 & 0 & A_{12}\\0 &B_{11} & B_{12} & 0\\ 0 & B_{21} & B_{22} & 0\\ A_{21} & 0 & 0 & A_{22}
    \end{pmatrix},
\end{align}
acting on adjacent qubits (arranged in 1D architecture), where $\det{A} = \det{B}$.

A quantum state is called a fermionic Gaussian state, if it may be obtained by the action of a FLO unitary on the initial, all zero vacuum state. Every fermionic Gaussian statevector $\ket{\psi}$ is conveniently described (up to a phase) by its covariance matrix 
\begin{align}
    M_{jk} = -\frac{i}{2}\bra{\psi} c_jc_k -c_k c_j\ket{\psi}\ .
\end{align}
which is anti-symmetric and real (due to the anti-commutation of the Majorana operators). The vacuum state $\ket{0}$ has a covariance matrix
\begin{align}
    M = \bigoplus_{j=0}^{n-1} \begin{pmatrix} 0 & 1 \\ -1 & 0 \end{pmatrix} .
\end{align}
Also, if the correlation matrix of $\ket{\psi}$ is $M$, and $U$ is a FLO unitary, then the correlation matrix of $U\ket{\psi}$ is 
\begin{align}
    \tilde{M} = \phi(U) M \phi(U)^T.
\end{align}
A real matrix $M$ of appropriate dimension is the correlation matrix of a pure fermionic Gaussian state if and only if $M + M^T = 0$ and $MM^T = I$. If one, in addition, considers mixed fermionic Gaussian states, then the second condition is weakened to $MM^T \leq I$.


\subsection{Passive and anti-passive FLO unitaries}

Our phase-sensitive FLO simulator will be based on a particular decomposition of every $U\in\operatorname{FLO}(n)$ into elements belonging to one of the two subsets of $\operatorname{FLO}(n)$. In order to define them, and also for future convenience, we need some notation. First, let
$N=\sum_{j=0}^{n-1} (I - i c_{2j}c_{2j+1})$  be the number operator counting how many fermions there are in our system. Let $h$ be the $2n\times 2n$ complex matrix defined by $h_{kj} = \bra{0}\maj{j}\maj{k}\ket{0}$ and note that $N = \sum_{jk} h_{jk}\maj{j}\maj{k}$. By expanding the Majorana operators in terms of creation and annihilation operators, one can demonstrate that
\begin{align}
  h = I_{2n} + I_n \otimes\sigma_y,\label{eqn:h-splits-into-ident-and-omega}
\end{align}
implying that the standard symplectic form
\begin{align}
  \Omega = I_n \otimes i\sigma_y, \label{eqn:standard-symplectic-form}
\end{align}
is a complex linear combination of $h$ and the $2n\times 2n$ identity matrix. Finally note that $i\sigma_y$ and $I_n \otimes i\sigma_y$ are real antisymmetric matrices.  

We then have the following definitions.
\begin{defi}[Passive FLO unitaries $\mathcal{K}$]
   The subgroup of \emph{passive FLO unitaries} $\mathcal{K}$ consists of all FLO unitaries with a generating matrix $\alpha$ satisfying $\com{\alpha}{\Omega} = 0$.
\end{defi}
\begin{defi}[Anti-passive FLO unitaries $\mathcal{P}$]
   The subset of \emph{anti-passive FLO unitaries} $\mathcal{P}$ consists of all FLO unitaries with a generating matrix $\alpha$ satisfying $\acom{\alpha}{\Omega} = 0$.
\end{defi}
Note that the above definition of passive fermionic linear optics is equivalent to the usual one, in which one defines this group as the unitaries of the form $U=\exp(i\sum_{k,l} X_{kl} a_k^\dagger a_l)$, for a Hermitian $n\times n$ matrix $X$. We summarise some properties of the subgroup $\mathcal{K}$ and subset $\mathcal{P}$ in the following lemmas, the proofs of which can be found in Appendix~\ref{app:properties}.

\begin{lem}[Properties of $\mathcal{K}$]\label{lem:props-of-k-type-flo}
  If $U_{t} = \exp{\left(\frac{t}{4}\sum_{jk}\alpha_{jk} \maj{j}\maj{k}\right)}$ is a one parameter group generated by some real antisymmetric matrix $\alpha$, then the following statements are equivalent:
  \begin{enumerate}
  \item $U_t \in \mathcal{K}$, $\forall t \in \mathbb{R}$.
  \item $\Omega\exp\left(t\alpha\right) \Omega^T = \exp{\left(t\alpha\right)} $, $\forall t \in \mathbb{R}$.
  \item $U_t N U_t^\dagger = N$, $\forall t \in \mathbb{R}$.
  \item $U_t \ket{0} \propto\ket{0}$, $\forall t \in \mathbb{R}$.
  \item $\alpha = A\otimes I + B\otimes i\sigma_y$ and the complex matrix $A + iB$ is anti-Hermitian, $(A + iB)^\dagger = -A - iB$.
  \item $\exp\left(t\alpha\right) = c_t\otimes I + s_t\otimes i\sigma_y$ and the $n\times n$ complex matrices $c_t + i s_t$ are unitary.
  \end{enumerate}
\end{lem}

\begin{lem}[Properties of $\mathcal{P}$]\label{lem:props-of-p-type}
  If $U_{t} = \exp{\left(\frac{t}{4}\sum_{jk}\beta_{jk} \maj{j}\maj{k}\right)}$ is a one parameter group generated by some real antisymmetric matrix $\beta$, then the following statements are equivalent:
  \begin{enumerate}
  \item $U_t \in \mathcal{P}$, $\forall t \in \mathbb{R}$.
  \item $\Omega \exp\left(t\beta \right)\Omega^{\dagger} = \exp\left(-t\beta \right)$, $\forall t\in\mathbb{R}$.
  \item $\beta = X\otimes\sigma_x + Z \otimes\sigma_z$ for some real, antisymmetric matrices $X$ and $Z$.
  \end{enumerate}
\end{lem}

Although anti-passive FLO unitaries do not form a group, using the following result (the proof of which can be found in Appendix~\ref{app:properties}), one can focus on a subset of $\mathcal{P}$ with a group structure.

\begin{lem}[Decomposition of $\mathcal{P}$]\label{lem:decomp-p-type}
  If $n$ is even, and $\beta$ is a $2n\times 2n$ real, antisymmetric matrix satisfying $\acom{\beta}{\Omega} = 0$, then there exists a real antisymmetric matrix $\alpha$ such that $\com{\alpha}{\Omega} = 0$ and
  \begin{align}
    e^{\alpha} \beta  e^{-\alpha}  =A
  \end{align}
  with
  \begin{equation}
    \label{eq:A_generator}
      A = \Lambda\otimes i\sigma_y\otimes \sigma_z
  \end{equation}
  for some real diagonal matrix $\Lambda$ of size $n/2$.
\end{lem}
If $n$ is odd, then an analogue of Lemma~\ref{lem:decomp-p-type} still holds but $A$ is padded by an additional row and column where every element is zero. Lemma~\ref{lem:decomp-p-type} suggests then a useful subset $\mathcal{A}$ of $\mathcal{P}$.
\begin{defi}[Commuting anti-passive FLO unitaries $\mathcal{A}$]
   The subgroup of \emph{commuting anti-passive FLO unitaries} $\mathcal{A}$ consists of all FLO unitaries with a generating matrix \changed{of the form of} $A$ from Eq.~\eqref{eq:A_generator}.
\end{defi}
All elements of $\mathcal{A}$ commute with each other, so it is easy to verify that $\mathcal{A}$ forms a subgroup of $\operatorname{FLO}(n)$ in contrast to $\mathcal{P}$, which is just a subset.


\section{Matchgate magic states and FLO extent}
In addition to the fermionic Gaussian states, we distinguish a class of states that we will call \emph{matchgate magic states}~\cite{PhysRevA.73.042313,deMelo2013,Oszmaniec2014}. Such states, through the use of gadgetization scheme, allow one to promote quantum circuits composed of FLO unitaries to universal quantum circuits. We note that, in order to be used in gadgetization schemes, a magic state must be \emph{fermionic} according to the definition of Ref.~\cite{Yoganathan-2019-magic-for-matchgates} i.e., be an eigenstate of the parity operator. Otherwise, it is not possible to move it to the correct position in the circuit and be consumed in a gadget~\cite{Yoganathan-2019-magic-for-matchgates}. In addition, all pure fermionic states of 1,2 or 3 qubits are Gaussian, as shown in Ref.~\cite{bravyi-2005-fermionic-product}. Thus, the simplest example of a magic state is a $4$ qubit entangled state. In particular, in this paper we will employ the following class of magic states: 
\begin{align}
  \ket{M_\theta} &= \frac{1}{2}\left(\ket{0000} + \ket{1100} + \ket{0011} + e^{i\theta}\ket{1111}\right).\label{eqn:Yoganathan-magic-states}
\end{align}
We call such states controlled-phase magic states, since they are consumed by the gadget shown in Fig.~\ref{fig:gadget1} to implement controlled-phase gates,
\begin{equation}
    C(\theta)=\ketbra{00}{00}+\ketbra{01}{01}+\ketbra{10}{10}+e^{i\theta}\ketbra{11}{11}.
\end{equation}

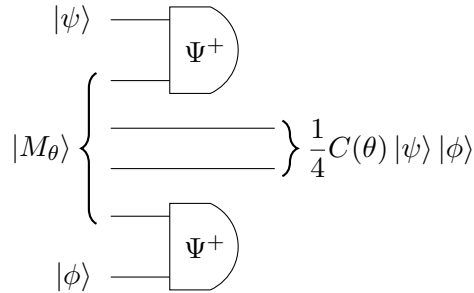
\begin{figure}
    \centering
    \begin{tikzpicture} [thick]
    \pgfdeclarelayer{background}
    \pgfdeclarelayer{foreground}
    \pgfsetlayers{background,main,foreground}
    \tikzstyle{operator} = [draw,fill=white,minimum size=1.5em] 
    \tikzstyle{phase} = [draw,fill,shape=circle,minimum size=0.2cm,inner sep=0pt]
    \tikzstyle{surround} = [fill=blue!10,draw=black,rounded corners=2mm]
    \tikzstyle{multiQubitMeasurement} = [fill=white,draw=black,shape=rounded rectangle,rounded rectangle left arc=none,rounded rectangle arc length=130]
    \matrix[row sep=0.5cm, column sep=1cm] (circuit) {
    \node[label=left:{$\ket{\psi}$}] (q1) {}; & 
    \node[phase] (M11) {}; &
    &[-0.3cm]
    \coordinate (end1); \\
    \node (q2) {}; & 
    \node[phase] (M12) {}; &
    &
    \coordinate (end2);\\
    \node (q3) {}; & &
    \coordinate (end3);\\
    \node (q4) {}; & &
    \coordinate (end4);\\
    \node (q5) {}; & 
    \node[phase] (M21) {}; &
    &[-0.3cm]
    \coordinate (end5); \\
    \node[label=left:{$\ket{\phi}$}] (q6) {};  &
    \node[phase] (M22) {}; &
    &[-0.3cm]
    \coordinate (end6); \\
    };
    
    \draw[decorate,decoration={brace,amplitude=0.2cm}]
        ($(end3)+(0.1cm,0.1cm)$)
        to node[midway,right] (bracket) {$\,\,\,\,\displaystyle \frac{1}{4}C(\theta)\ket{\psi}\ket{\phi}$}
        ($(end4)+(0.1cm,-0.1cm)$);
        
    \draw[decorate,decoration={brace,amplitude=0.2cm,mirror}]
        ($(q2)+(-0.1cm,0.1cm)$)
        to node[midway,left] (bracket) {$\displaystyle\ket{M_\theta}\,\,\,\,$}
        ($(q5)+(-0.1cm,-0.1cm)$);
        
    \begin{pgfonlayer}{background}
        \draw[] (q1) -- (M11)  (q2) -- (M12) (q3) -- (end3) (q4) -- (end4) (q5) -- (M21) (q6) -- (M22); 
    \end{pgfonlayer}
    \begin{pgfonlayer}{foreground}
        \node[multiQubitMeasurement, fit = (M11) (M12)] (M1) {};
        \node[right=-0.15cm of M1.base] {$\Psi^+$};
        \node[multiQubitMeasurement, fit = (M21) (M22)] (M2) {};
        \node[right=-0.15cm of M2.base] {$\Psi^+$};
    \end{pgfonlayer}
    \end{tikzpicture}
    \caption{\label{fig:gadget1} \textbf{Gadget for $C(\theta)$.} The gadget described in Ref.~\cite{Yoganathan-2019-magic-for-matchgates} to implement a controlled-phase gate. Each gadget is adapted for the use of classical simulation, where correction gates are uneccessary, and one can instead simply project onto the ``success'' outcome of the measurement with $\ket{\Psi^+} = (\ket{00} + \ket{11})/\sqrt{2}$. There are $16$ possible combined outcomes of the two $4$ outcome measurements for each gadget. Since each outcome occurs with equal probability projecting onto the the all zero outcome divides the norm of the state by $4 = \sqrt{16}$ in addition to implementing the unitary $C(\theta)$.}
\end{figure}

The following notion will be used to quantify how far from a fermionic Gaussian state a given state is.

\begin{defi}[FLO extent]
\label{def:extent}
  For any pure state $\ket{\psi}$, the \emph{FLO extent} $\xi(\ket{\psi})$ is the infimum of $\norm{a}_1^2$ over all finite dimensional complex vectors $a$ such that
  \begin{align}
    \ket{\psi} = \sum_j a_j \ket{s_j},\label{eqn:def-extent-decompostion}
  \end{align}
  for some set of normalized fermionic Gaussian states labelled $\ket{s_j}$.
\end{defi}
The notion of extent was introduced in Ref.~\cite{bravyi-2019} in the context of simulation algorithms based on Clifford/stabilizer subtheory of quantum mechanics, and our definition~\ref{def:extent} is a simple generalisation to the case of fermionic linear optics. In papers in the Clifford/stabilizer literature, such as Refs.~\cite{bravyi-2019,PRXQuantum.3.020361,PRXQuantum.2.010345}, the extent has emerged as a key quantity bounding the runtime of simulation algorithms in terms of the ``nonclassicality'' of the computation being simulated. In Ref.~\cite{bravyi-2019}, the exponential component of the runtime is controlled by a quantity known as the approximate stabilizer rank, which in turn is upper-bounded by the extent; while the runtime of the key algorithm of Ref.~\cite{bravyi-2019,PRXQuantum.3.020361} is found in terms of the extent directly. We emphasise that the quantity we use here is the fermionic linear optical or \emph{FLO extent}, while that studied in the literature we cite above is the \emph{stabilizer extent}. The difference lies in the set of states forming the decomposition in Eq.~\eqref{eqn:def-extent-decompostion}. While some properties (such as sub-multiplicativity) are identical, it is not obvious how the two differ.

Since a key step of our simulation algorithm is based on the decomposition of tensor products of matchgate magic states from Eq.~\eqref{eqn:Yoganathan-magic-states} as superpositions of fermionic Gaussian states, we will now present some basic properties of the FLO extent, together with a calculation of the extent of $\ket{M_\theta}$. First, there is the easy to show \emph{sub-multiplicative} bound
\begin{align}
  \ext{\bigotimes_j \ket{\psi_j}} \leq \prod_j \ext{\ket{\psi_j}}.
\end{align}
Next, using standard methods of convex optimization (also employed in Ref.~\cite{bravyi-2019}), we obtain the following result.
\begin{lem}[Dual problem for the extent]\label{lem:dual-problem-for-extent}
  For any pure state $\ket{\psi}$ of $n$ qubits, we have
  \begin{align}
    \ext{\ket{\psi}} = \max_\omega \frac{\abs{\braket{\omega}{\psi}}^2}{\fid{\ket{\omega}}},\label{eqn:extent-dual-problem}
  \end{align}
  where the max is over all normalized $n$-qubit pure states $\ket{\omega}$, and
  \begin{align}
    \fid{\ket{\omega}} &= \max_{U\in\text{FLO}(n)} \abs{\bra{\omega}U\ket{0}}^2
  \end{align}
  is the fermionic linear optical \emph{fidelity}, i.e., the greatest fidelity between $\ket{\omega}$ and any fermionic Gaussian state.
\end{lem}
Also, by adapting the method used in proving Proposition~$2$ of Ref.~\cite{bravyi-2019}, we obtain the following lemma, the proof of which can also be found in Appendix~\ref{app:extent}.
\begin{lem}\label{lem:extent-for-F3}
  If $\ket{\psi}$ is a pure state of $n$ qubits such that 
  \begin{align}
    \ketbra{\psi}{\psi} &= \frac{1}{N}\sum_{j= 1}^{N}  V_j,
  \end{align}
  where each $V_j$ belongs to $\operatorname{FLO}(n)$, then
  \begin{align}
    \ext{\ket{\psi}} &= \frac{1}{\fid{\ket{\psi}}}.
  \end{align}
\end{lem}
Although the assumptions of Lemma~\ref{lem:extent-for-F3} may seem unmotivated, in fact it covers a very important special case, i.e., augmenting FLO operations with \emph{stabilizer state} inputs. While neither fermionic linear optics nor Clifford/stabilizer quantum mechanics is universal for quantum computation, their combination is. In particular, it is well-known~\cite{PhysRevA.73.042313} that augmenting fermionic linear optics with either the SWAP or controlled-Z gate suffices for universal quantum computation. We note that the magic state shown in Eq.~\eqref{eqn:Yoganathan-magic-states} falls into this category exactly when $\theta$ is a multiple of $\pi$.

Concerning the extent of the particularly important family of magic states $\ket{M_\theta}$, we have the following result. 

\begin{thm}[Decomposition of magic states for controlled-phase gates]\label{thm:decomposition-of-magic-states}
The magic state $\ket{M_\theta}$ from Eq.~\eqref{eqn:Yoganathan-magic-states} may be decomposed as a sum of two fermionic Gaussian states
\begin{align}
  \ket{M_\theta} = &\cos\left(\frac{\theta}{4}\right)\ket{A(\theta)} + i\sin\left(\frac{\theta}{4}\right)\ket{B(\theta)},\label{eqn:magic-state-orthogonal-decomposition}
\end{align}
where 
\begin{align}
    \ket{A(\theta)} &= \frac{1}{2}\left(e^{-i\frac{\theta}{4}}\ket{0000} + e^{i\frac{\theta}{4}}\ket{0011} + e^{i\frac{\theta}{4}}\ket{1100} + e^{i\frac{3\theta}{4}}\ket{1111}\right),\\
    \ket{B(\theta)} &=\frac{1}{2}\left(e^{-i\frac{\theta}{4}}\ket{0000} - e^{i\frac{\theta}{4}}\ket{0011} -e^{i\frac{\theta}{4}}\ket{1100} + e^{i\frac{3\theta}{4}}\ket{1111}\right).
\end{align}
Explicit calculation using this decomposition provides a bound on the extent
\begin{align}
  \ext{\ket{M_\theta}} &\leq \left(\abs{\cos\left(\frac{\theta}{4}\right)} + \abs{\sin\left(\frac{\theta}{4}\right)}\right)^2 = 1 + \abs{\sin\left(\frac{\theta}{2}\right)}.
\end{align}
In fact this bound is tight, and for all our magic states 
\begin{align}
  \ext{\ket{M_\theta}} =  1 + \abs{\sin\left(\frac{\theta}{2}\right)}.
\end{align}
\end{thm}
\begin{proof}
    For technical reasons, we prove the result not for $\ket{M_\theta}$ directly, but for the statevector
    \begin{align}
        \ket{\tilde{M}_\theta} &= e^{-i\frac{\pi}{4}} e^{-i\frac{\theta}{4}} \ket{M_\theta}.
    \end{align}
    Multiplying both sides of Eq.~\eqref{eqn:magic-state-orthogonal-decomposition} by the phase $e^{-i\frac{\pi}{4}} e^{-i\frac{\theta}{4}}$ we obtain 
    \begin{align}
        \ket{\tilde{M}_\theta} = &\cos\left(\frac{\theta}{4}\right)\ket{\tilde{A}(\theta)} + \sin\left(\frac{\theta}{4}\right)\ket{\tilde{B}(\theta)},\label{eqn:magic-state-orthogonal-decomposition-2}
    \end{align}
    where
    \begin{align}
    \ket{\tilde{A}(\theta)} &= \frac{1}{2}e^{-i\frac{\pi}{4}} \left(e^{-i\frac{\theta}{2}}\ket{0000} + \ket{0011} + \ket{1100} + e^{i\frac{\theta}{2}}\ket{1111}\right),\\
    \ket{\tilde{B}(\theta)} &=\frac{1}{2}e^{i\frac{\pi}{4}} \left(e^{-i\frac{\theta}{2}}\ket{0000} - \ket{0011} - \ket{1100} + e^{i\frac{\theta}{2}}\ket{1111}\right),
\end{align}
    are easily shown to be fermionic Gaussian states via the criteron shown in Lemma 2 of Ref.~\cite{bravyi-2005-lagrangian-rep}.
    Let $T$ be the anti-unitary operator\footnote{Note that anti-unitary operators are real-linear but \emph{not} complex-linear and in particular they apply complex conjugation to non-real coefficients.} acting the Hilbert space of $4$ qubits, satisfying $kTk^\dagger = T$ for all $k\in\mathrm{FLO}$. The existence of this anti-unitary is guaranteed by the results of Refs.~\cite{PhysRevA.80.022319,Oszmaniec2014,Oszmaniec_2012}. It is called $\theta$ in those references but we rename it here to avoid confusion with the angle $\theta$. Since $T k = k T$ for all FLO unitaries, in particular $Tc_j c_k = c_j c_k T$ for all Majorana fermion operators $c_j$ and $c_k$. Then, expressing the creation and anihilation operators in terms of the Majorana operators we obtain $a_j a_k T = T a_j^\dagger a_k^\dagger$ and $a_j^\dagger a_k^\dagger T = T a_j a_k$, which implies $T\ket{0000} \propto \ket{1111}$. Setting $T\ket{0000} = \ket{1111}$ to fix the phase-ambiguity in $T$, we obtain
    \begin{subequations}
    \begin{align}
        T \ket{0000} &= \ket{1111} & T\ket{1111} &=\ket{0000} \\
        T\ket{0011} &= -\ket{1100} &T\ket{1100} &=-\ket{0011} \\
        T\ket{0101} &= \ket{1010}&T\ket{1010} &=\ket{0101} \\
        T\ket{1001} &= -\ket{0110} &T\ket{0110} &=-\ket{1001}.
    \end{align}
    \end{subequations}
    Note that this differs from the corresponding expression in Ref.~\cite{Oszmaniec2014} by a sign flip in the expressions for $T\ket{0101}$ and $T\ket{0101}$. With this notation, Eq.~\eqref{eqn:magic-state-orthogonal-decomposition-2} may be rewritten as
    \begin{align}
        \ket{\tilde{M}_\theta} = &\cos\left(\frac{\theta}{4}\right)\ket{\tilde{A}(\theta)} + \sin\left(\frac{\theta}{4}\right)T\ket{\tilde{A}(\theta)},
    \end{align}
    and, in particular, has exactly the same form as equation~(12) of Ref.~\cite{Oszmaniec2014}. The following basis of vectors invariant under $T$ is relevant for our proof
    \begin{subequations}
    \begin{align}
        \ket{\eta_1} &= \frac{1}{\sqrt{2}}\left(\ket{0000} + \ket{1111}\right), &
        \ket{\eta_2} &= \frac{i}{\sqrt{2}}\left(\ket{0000} - \ket{1111}\right),\\
        \ket{\eta_3} &= \frac{1}{\sqrt{2}}\left(\ket{0011} - \ket{1100}\right),&
        \ket{\eta_4} &= \frac{i}{\sqrt{2}}\left(\ket{0011} + \ket{1100}\right),\\
        \ket{\eta_5} &= \frac{i}{\sqrt{2}}\left(\ket{0101} - \ket{1010}\right), &
        \ket{\eta_6} &= \frac{1}{\sqrt{2}}\left(\ket{0101} + \ket{1010}\right),\\
        \ket{\eta_7} &= \frac{1}{\sqrt{2}}\left(\ket{1001} - \ket{0110}\right),&
        \ket{\eta_8} &= \frac{i}{\sqrt{2}}\left(\ket{1001} + \ket{0110}\right).
    \end{align}
    \end{subequations}
    We will now employ Lemma~\ref{lem:dual-problem-for-extent} with the witness state 
    \begin{align}
        \ket{\omega_\theta} &= \frac{1}{\sqrt{2}}\left(\ket{\tilde{A}(\theta)} + \ket{\tilde{B}(\theta)}\right).
    \end{align}
    Using the results of Ref.~\cite{Oszmaniec2014}, the FLO-fidelity of $\ket{\omega_\theta}$ is easily evaluated. First, note that $\ket{\omega_\theta}$ is a \emph{real} statevector, in the sense of having real coefficients when expressed in the $\ket{\eta}$ basis. At the same time, due to Eq.~(11) of Ref.~\cite{Oszmaniec2014}, every fermionic-Gaussian state on 4 fermionic modes has the form
    \begin{align}
        \ket{\psi} &= \frac{1}{\sqrt{2}}\left(\ket{\psi_1} + i \ket{\psi_2}\right),
    \end{align}
    for two orthonormal statevectors $\ket{\psi_1}$ and $\ket{\psi_2}$, which \emph{also} each have real coefficients when expressed in the $\ket{\eta}$ basis. We can then evaluate
    \begin{align}
        \left|\braket{\omega_\theta}{\psi}\right|^2 &= \frac{1}{2}\left(\braket{\omega_\theta}{\psi_1}^2 +  \braket{\omega_\theta}{\psi_2}^2 \right),
    \end{align}
    since all three statevectors are real in the same basis. Further, since the states are normalized and the $\ket{\psi_j}$ are orthogonal, it follows that for any $\ket{\psi}$ the overlap with $\ket{\omega_\theta}$ is bounded by $1/2$. In order to show the FLO-Fidelity of $\ket{\omega_\theta}$ is exactly $1/2$, it remains only to show that there is some FLO state $\ket{\psi}$ such that
    \begin{align}
        \left|\braket{\omega_\theta}{\psi}\right|^2 = \frac{1}{2}.
    \end{align}
    However, it is easy to see that either of $\ket{\tilde{A}_\theta}$ or $\ket{\tilde{B}_\theta}$ suffices. 
    
    Now, we have
    \begin{align}
        \fid{\ket{\omega_\theta}} &= \frac{1}{2},
    \end{align}
    and we employ the dual problem from Lemma~\ref{lem:dual-problem-for-extent} to obtain
    \begin{align}
        \ext{\ket{\tilde{M}_\theta}} &= \max_\omega \frac{\abs{\braket{\omega}{\tilde{M}_\theta}}^2}
        {\fid{\ket{\omega}}}\\ &\geq \frac{\abs{\braket{\omega_\theta}{\tilde{M}_\theta}}^2}{\fid{\ket{\omega_\theta}}}\\
        &= 2 \abs{\braket{\omega_\theta}{\tilde{M}_\theta}}^2\\
        &= \left(\cos\left(\frac{\theta}{4}\right) + \sin\left(\frac{\theta}{4}\right)\right)^2
    \end{align}
\end{proof}

It is interesting that the structures examined in Ref.~\cite{Oszmaniec2014} force the extent-optimal decomposition of our magic states to be an orthogonal decomposition. This in stark contrast to the case of optimal magic state decompositions in the Clifford/stabilizer case, where single-qubit $\ket{T}$ states are optimally expressed as a sum of two states coming from a pair of mutually unbiased bases. 

Finally, note that due to the results of Ref.~\cite{Oszmaniec2014}, the witness states $\ket{\omega_\theta}$ appearing in the proof of Theorem~\ref{thm:decomposition-of-magic-states} are all in the FLO-orbit of the state \changed{$\ket{a_8} = \ket{\eta_1}$} defined in Ref.~\cite{PhysRevA.73.042313}. That is, they are all of the form $U\ket{a_8}$ for some FLO unitary $U$. Applying the same reasoning used above to a tensor-product of magic states with the obvious product witness provides the bound
\begin{align}
    \ext{\prod_{j=1}^m \ket{M_{\theta_j}}} \geq \frac{2^{-m}}{\fid{\ket{a_8}^{\otimes m}}}\prod_{j=1}^m \left(\cos\left(\frac{\theta_j}{4}\right) + \sin\left(\frac{\theta_j}{4}\right)\right),
\end{align}
which, under the plausible conjecture that the FLO-fidelity is multiplicative for products of $\ket{a_8}$ states, would imply that the FLO extent is multiplicative for the class of magic states we consider
\begin{align}
    \ext{\prod_{j=1}^m \ket{M_{\theta_j}}} = \prod_{j=1}^m \left(\cos\left(\frac{\theta_j}{4}\right) + \sin\left(\frac{\theta_j}{4}\right)\right).
\end{align}
In Appendix~\ref{app:fidelity-of-a8-pair} we show that $\fid{\ket{a_8}^{\otimes 2}} = \fid{\ket{a_8}}^2 = 1/4$, proving this conjecture for $m=2$ $4$-qubit magic states. The conjecture was subsequently proved in full by one of the authors in Ref.~\cite{reardonsmith2024fermioniclinearopticalextent}, where it was shown that the fermionic linear optical fidelity is multiplicative for arbitrary tensor-products of $\ket{a_8}$, and the extent is multiplicative for arbitrary tensor products of $4$ qubit parity eigenstates.

\section{Phase-sensitive FLO simulation}
\label{sec:phase-sensitive-flo-sim}

In this section, we will describe our phase-sensitive simulation algorithm for fermionic linear optics. We will first give a classical efficient description of every pure fermionic Gaussian state. Then, we will explain how to efficiently update this description under the evolution induced by FLO unitaries, in a way that allows one to preserve the global phase information.


\subsection{Representation of fermionic Gaussian states}

A key result we will employ is the following lemma, which is an example of a Cartan's (or ``$KAK$'' type) decomposition.
\begin{lem}\label{lem:kak-FLO-unitaries}
  For any FLO unitary $U$ there exist two passive FLO unitaries $K_1, K_2\in\mathcal{K}$ and one anti-passive FLO unitary $A\in\mathcal{A}$ such that
  \begin{align}
    U = K_1 A K_2.
  \end{align}
\end{lem}
The proof of the above decomposition can be found in Appendix~\ref{app:Cartans-decomp} and a practical (runtime $\order{n^3}$) method for obtaining the decomposition in Lemma~\ref{lem:app-recovering-phases}. There is a deep and powerful theory of such decompositions underlying much of the structure of Lie groups and symmetric spaces~\cite{HelgasonGifGeomLieGroupAndSym,KnappLieGroupsBeyond}. Here, however, we will use almost none of this theory, and instead rely on explicit matrix computations. 

Note that $A\in\mathcal{A}$ is a product of $\order{n}$ commuting operators of the form
\begin{align}
    \exp(\lambda c_j c_k)\label{eqn:A-type-FLO-unitary},
\end{align}
which we will refer to as \emph{elementary FLO unitaries}. Such gates, in combination with passive FLO unitaries, therefore generate the full group $\operatorname{FLO}(n)$. We make this argument explicit in Appendix~\ref{app:recovering-phases}, where in Lemma~\ref{lem:app-recovering-phases} we provide a polynomial-time algorithm to compute the $KAK$ decomposition from a generating matrix of a FLO unitary $U$. The division into these two categories is relevant because the phase-sensitive FLO simulation algorithms we will develop in this section have a similar form to the ``CH-form'', introduced in Ref.~\cite{bravyi-2019} for phase-sensitive simulation of stabilizer quantum mechanics. There, the elementary Clifford unitaries split into $CX$, $S$ and $CZ$ which preserve the $\ket{0}$ state, and the Hadamard gate $H$ which does not. The main result upon which our simulation algorithms depend is given by the following lemma.
\begin{lem}\label{lem:FLO-BM-form}
  Any fermionic Gaussian state $\ket{\psi}$ of $n$ qubits may be written as
  \begin{align}
    \ket{\psi} &= \omega K \prod_{j=0}^{\changed{\frac{n}{2}-1}} \exp\left(\lambda_j c_{4j} c_{4j+2}\right)\ket{0},
  \end{align}
  where $\omega\in\mathbb{C}$ is a phase factor, $K\in\mathcal K$ is a passive FLO unitary (so that $K\ket{0} \propto \ket{0}$), and $\lambda_j\in\mathbb{R}$.
\end{lem}
\changed{Note that as stated and proved here this lemma applies only to systems with an even number of qubits, so $\frac{n}{2}$ is an integer. This does not lose any generality as an extra unused qubit may always be added to a circuit. In fact adding an extra qubit is unnecessary and our results generalise to odd qubit numbers with some minor but notationally cumbersome adjustments.} 
\begin{proof} By definition, any fermionic Gaussian state is obtained by the application of some FLO unitary $U$ to the vacuum state. Applying Lemma~\ref{lem:kak-FLO-unitaries} to that unitary, we obtain
\begin{align}
  \ket{\psi} &= U \ket{0}\\
             &= K_1 \exp\left(\sum_{j=0}^{\changed{\frac{n}{2}-1}} \lambda_j \left(c_{4j} c_{4j+2} - c_{4j+1} c_{4j+3}\right)\right) K_2 \ket{0}\\
             &= \omega K_1 \prod_{j=0}^{\changed{\frac{n}{2}-1}} \exp\left(2 \lambda_j c_{4j} c_{4j+2}\right)\ket{0},
\end{align}
where we have used the fact that $K_2\ket{0} = \omega\ket{0}$ for some complex $\omega$, and employed the relation 
\begin{equation}
    \exp\left(\lambda_j \left(c_{4j} c_{4j+2} - c_{4j+1} c_{4j+3}\right)\right)\ket{0} = \exp\left(2\lambda_j c_{4j} c_{4j+2}\right)\ket{0}.    
\end{equation}

\end{proof}

Since we will be interested in classical simulation algorithms, we must be rather explicit about the data we use to describe a fermionic Gaussian state. We summarise this data in the following definition.

\begin{defi}[Classical description of a fermionic Gaussian state]

\label{def:fermionic-gaussian-classical-data}
  A fermionic Gaussian statevector $\ket{\psi}$ on $n$ qubits is specified by its \emph{classical description} $\Gamma = \left(\omega, R, a, \lambda\right)$, where $R$ is a real $2n\times 2n$ special orthogonal matrix, $a$ and $\omega$ are complex numbers, and $\lambda$ is a length $\frac{n}{2}$ vector of real numbers. Explicitly,
  \begin{align}
    \ket{\psi} &= \omega K \exp\left(\sum_{j=0}^{\frac{n}{2}-1} \lambda_j c_{4j} c_{4j+2}\right) \ket{0},
  \end{align}
  where
  \begin{align}
    K\ket{0} &= a\ket{0}\label{eqn:def-flo-chform-a},\\
    K c_j K^\dagger &= \sum_{k=0}^{2n-1} R_{jk} c_k.\label{eqn:def-flo-chform-r-matrix}
  \end{align}
\end{defi}




\subsection{Update rules}

For classical simulation, it is essential for us to have efficient subroutines to update the data defining a state $\ket{\psi}$ upon application of a FLO unitary $K\in\mathcal{K}$ or $\exp(\lambda c_j c_k)$. For passive FLO unitaries these update rules are rather simple.
\begin{updaterule}[Application of a passive FLO unitary]
  ~\\If $\Gamma = \left(\omega, R, a, \lambda\right)$ is a classical description of a fermionic Gaussian state $\ket{\psi}$, and $U$ is a passive FLO unitary satisfying
  \begin{subequations}
  \begin{align}
    U\ket{0} &= b\ket{0},\\
    U c_j U^\dagger &= \sum_{k=0}^{2n-1} S_{jk} c_k,
  \end{align}
  \end{subequations}
  then the data describing the state $U\ket{\psi}$ is
  \begin{align}
    S \cdot \Gamma = \left(\omega, RS, ab, \lambda \right).
  \end{align}
\end{updaterule}
The proof is essentially by inspection, upon recalling Eqs.~\eqref{eqn:def-flo-chform-a} and~\eqref{eqn:def-flo-chform-r-matrix}. Unfortunately, the update rule for FLO unitaries of the form $U = \exp\left(\lambda c_j c_k\right)$ is somewhat more involved. We summarise the steps below. 
\changed{
\begin{updaterule}[Application of a commuting anti-passive FLO unitary]
  ~\\If $\Gamma = \left(\omega, R, a, \lambda\right)$ is a classical description of a fermionic Gaussian state $\ket{\psi}$, $U = \prod_{j=0}^{\frac{n}{2}-1}\exp\left(\frac{\mu_j}{2} (c_{4j} c_{4j+2} -c_{4j+1} c_{4j+3})\right)$ is an anti-passive FLO unitary and we define $A = \prod_{j=0}^{\frac{n}{2}-1}\exp\left(\frac{\lambda_j}{2} (c_{4j} c_{4j+2} -c_{4j+1} c_{4j+3})\right)$, then the data describing the state $U\ket{\psi}$ may be computed using the following steps:
  \begin{enumerate}
  \item Compute the special orthogonal matrices $Q_1= \phi(U)$ and $Q_2 = \phi(A)$.
  \item Perform the multiplication in the special orthogonal group to obtain a matrix $Q_2 R Q_1 = \phi(UKA)$
  \item Compute  the real anti-symmetric matrix $\beta = \log(Q_2 R Q_1 )$. For example by using the real Schur decomposition to block-diagonalize $Q_2 R Q_1 $.
  \item Now, since the kernel of the covering map $\phi$ from Eq.~\eqref{eq:phi} is $\pm I$, we have
    \begin{align}
      U K A= \pm \exp\left(\frac{1}{4}\sum_{j,k=0}^{2n-1}\beta_{jk} c_j c_k\right).
    \end{align}
  \item We correct the unknown phase $\pm 1$ employing the methods described in Appendix~\ref{app:recovering-phases}, Lemma~\ref{lem:phase-sensitive-anti-passive}. In summary, we provide an efficient classical algorithm which finds a fermionic Gaussian state $\ket{\xi}$ such that the matrix element $m = \bra{\xi} \exp\left(\frac{1}{4}\sum_{j,k=0}^{2n-1}\beta_{jk} c_j c_k\right)\ket{\xi}$ has absolute value $1$. Then, we compute $m$ and $\bra{\xi}U K A\ket{\xi}$, and compare them to see if they differ by a factor of minus one or are equal.
  \item Apply the KAK decomposition method described in Lemma~\ref{lem:app-recovering-phases} to $\exp\left(\frac{1}{4}\sum_{j,k=0}^{2n-1}\beta_{jk} c_j c_k\right)$ to obtain passive FLO unitaries $\tilde{K}$ and $\hat{K}$ and an anti-passive $\tilde{A}$ such that $\exp\left(\frac{1}{4}\sum_{j,k=0}^{2n-1}\beta_{jk} c_j c_k\right) = \tilde{K}\tilde{A}\hat{K}$.
  \item Apply $\tilde{K}\tilde{A}\hat{K}$ to the vacuum state $\ket{0}$ to obtain $U\ket{\psi} = \tilde{\omega}\tilde{K}\tilde{A}\ket{0}$, if necessary combining the power of $-1$ found in step 5 with $\tilde{\omega}$.
  \end{enumerate}
\end{updaterule}
}
Of course, knowing the update rule for passive and commuting anti-passive FLO unitaries directly yields the following update rule for general FLO unitaries.
\changed{
\begin{updaterule}[Application of general FLO unitaries]
  ~\\Since every FLO unitary may be KAK-decomposed into a product of two passive FLO unitaries and a commuting anti-passive FLO unitary, applying a general FLO unitary may be acomplished by applying rule 1 twice and rule 2 once.
\end{updaterule}
}
Finally, we also provide a rule for calculating computational basis inner products.

\begin{updaterule}[Computational basis inner products]
  ~\\If $\Gamma = \left(\omega, R, a, \lambda\right)$ is a classical description of a fermionic Gaussian state $\ket{\psi}$, $x$ is a length $2n$ binary string, and $c(x) = \prod_j c_j^{x_j}$, then the inner product $\bra{0}c(x)\ket{\psi}$ may be computed using the steps:
  \begin{enumerate}
  \item Write
    \begin{align}
      \bra{0}c(x)\ket{\psi} &= \omega \tr\left( \ketbra{0}{0} c(x) K A \right) = \omega \bra{0}K\ket{0} \tr\left( (c(x) \ketbra{0}{0})\cdot(c(x) K A K^\dagger c(x)^\dagger)\right).
    \end{align}
  \item Rewrite the trace as the Pfaffian of an $\order{n} \times \order{n}$ matrix using the methods shown in Appendix~\ref{app:recovering-phases}.
  \item Compute the Pfaffian using standard methods.
  \end{enumerate}
\end{updaterule}

Based on the above update rules and the proofs and methods presented in Appendix~\ref{app:recovering-phases}, we can summarise the performance of our phase-sensitive classical simulation algorithms as follows. Given a classical description of a fermionic Gaussian state $\ket{\psi}$, we provide algorithms to compute a new classical definition of:
  \begin{itemize}
  \item The state $K\ket{\psi}$ for a passive FLO unitary $K$ -- in time $\order{n^3}$.
  \item The state \changed{$\exp\left(\sum_j \mu_j (c_{4j} c_{4j+2} -c_{4j+1} c_{4j+3}) \right) \ket{\psi}$} -- in time $ \order{n^3}$.
  \item The (subnormalized) state  $a_j a_j^\dagger\ket{\psi}$ -- in time $\order{n^3}$.
  \item The inner product $\braket{0}{\psi}$ -- time $\order{n}$.
  \item The inner product $\bra{0}\left(\prod_j c_j^{x_j}\right)\ket{\psi}$ for some binary vector $x$ -- in time $\order{n^3}$.
  \end{itemize}


\section{Simulating universal quantum circuits}
\label{sec:simulating-universal-circuits}

In this section, we will present our simulation algorithm that allows one for the estimation of Born-rule probabilities for universal quantum circuits composed of FLO unitaries (matchgates) and controlled phase gates. \changed{Note that restricting to circuits consisting of these gates implies no loss of generality. It was shown in Ref.~\cite{josaMatchgate2008} that circuits consisting of matchgates and swap gates are universal, indeed any
BQP algorithm can be simulated by a poly-sized circuit of matchgates and swap gates. Since a swap gate is equal to a controlled phase-gate with phase equal to $\pi$ multiplied by the matchgate $G(\sigma_Z,\sigma_X)$, this result directly applies to our gateset.} The main idea behind the algorithm is as follows. First, we will use gadgetization in order to replace the original circuit by a FLO circuit with non-FLO measurements post-selecting on the magic state $\ket{M_\theta}$ at the end. Then, we will employ our phase-sensitive FLO simulator to evolve the initial state to the final pre-measurement state. Next, we will decompose each $\ket{M_\theta}$ into a superposition of fermionic Gaussian states. Finally, using sampling methods based on the Hoeffding's inequality, we will estimate the desired probability.

In what follows, we will adapt the gadget described in Ref.~\cite{Yoganathan-2019-magic-for-matchgates}, although the method is similar to the one given earlier in Ref.~\cite{bravyi-kitaev-2000}. 
More precisely, 
we will reverse the gadget of Ref.~\cite{Yoganathan-2019-magic-for-matchgates} (presented in Fig.~\ref{fig:gadget1}) so that an input fermionic Gaussian state enters each gadget, and a non-FLO measurement is applied to implement each controlled-phase gate (see Fig.~\ref{fig:gadget2}). Experience from classical simulators based on Clifford/stabilizer quantum mechanics, e.g., Ref.~\cite{PRXQuantum.3.020361}, shows that reversing the gadgetization can provide significant improvements to the polynomial components of the runtime. This formulation is more efficient, as it allows us to apply the full gadgetized circuit to an initial (FLO) vacuum state, and then post-select on the non-FLO states at the end. Practically this means that the exponential component of the runtime does not have any prefactor proportional to the number of gates in the circuit.

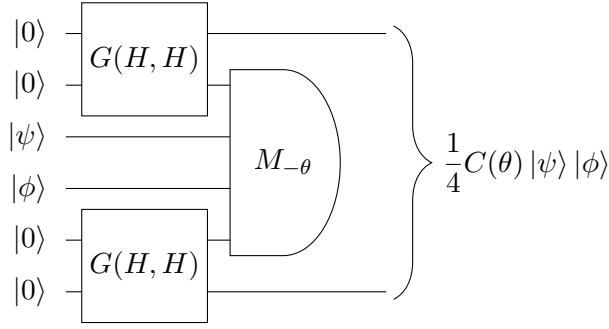
\begin{figure}
  \centering
\pgfdeclarelayer{background}
\pgfdeclarelayer{foreground}
\pgfsetlayers{background,main,foreground}
        \begin{tikzpicture}
    \tikzstyle{operator} = [draw,fill=white,minimum size=1.5em] 
    \tikzstyle{phase} = [draw,fill,shape=circle,minimum size=0.2cm,inner sep=0pt]
    \tikzstyle{multiQubitGate} = [fill=white,draw=black,shape=rectangle, minimum height=1.5cm, inner sep=1mm]
    \tikzstyle{singleQubitMeasurement} = [fill=white,draw=black,shape=rounded rectangle,rounded rectangle left arc=none,rounded rectangle arc length=130,inner sep=0.1cm, outer sep=0cm]
    %
    \matrix[row sep=0.3cm, column sep=1cm] (circuit) {
    \node[label=left:{$\ket{0}$}] (q1) {}; & 
    \node[phase] (G11) {}; & &
    \node (skip) {}; &
    \coordinate (end1); \\
    \node[label=left:{$\ket{0}$}] (q2) {}; & 
    \node[phase] (G12) {}; &  
    \node (M2) {}; &
    \coordinate (end2);\\
    \node[label=left:{$\ket{\psi}$}] (q3) {}; & 
    &
    \coordinate (end3);\\
    \node[label=left:{$\ket{\phi}$}] (q4){}; & &
    &
    \coordinate (end4);\\
    \node[label=left:{$\ket{0}$}] (q5) {}; & 
    \node[phase] (G21) {}; & 
    \node(M3) {}; &
    \coordinate (end5); \\
    \node[label=left:{$\ket{0}$}] (q6) {};  &
    \node[phase] (G22) {}; &
    \node (end7) {}; & &
    \coordinate (end6); \\
    };
    
     \begin{pgfonlayer}{foreground}
       \coordinate (rectbottomleft) at ($(M2)!1.1!(M3)$);
       \coordinate (recttopleft) at ($(M3)!1.1!(M2)$);
       \coordinate (recttopright) at ($(recttopleft) + (0.7cm, 0)$);
       \coordinate (rectbottomright) at ($(rectbottomleft) + (0.7cm, 0)$);
       \newdimen\mydim
       \pgfextracty{\mydim}{\pgfpointscale{0.5}{\pgfpointdiff{\pgfpointanchor{rectbottomright}{center}}{\pgfpointanchor{recttopright}{center}}}}
       \draw[fill=white, draw=black] (recttopleft) -- (rectbottomleft) -- (rectbottomright) arc[start angle=270, delta angle=180, x radius=.75cm, y radius=\mydim] -- cycle;
       \coordinate (rectmidleft) at ($(M2)!0.5!(M3)$);       
       \node (bigmeasurementlabel) at ($(rectmidleft) + (0.7cm, 0)$) {$\displaystyle{M_{-\theta}}$};
    \end{pgfonlayer}
    \draw[decorate,decoration={brace,amplitude=0.5cm}]
        ($(end1)+(0.1cm,0.1cm)$)
        to node[midway,right] (bracket) {$\,\,\,\,\,\,\,\,\,\displaystyle \frac{1}{4}C(\theta)\ket{\psi}\ket{\phi}$}
        ($(end6)+(0.1cm,-0.1cm)$);
        
        
    \begin{pgfonlayer}{background}
        \draw[] (q1) -- (end1)  (q2) -- (end2) (q3) -- (end3) (q4) -- (end4) (q5) -- (end5) (q6) -- (end6); 
    \end{pgfonlayer}
    \begin{pgfonlayer}{foreground}
        \node[multiQubitGate, at = ($(G11)!0.5!(G12)$)] (G1) {$G(H,H)$};
        \node[multiQubitGate, at = ($(G21)!0.5!(G22)$)] (G2) {$G(H,H)$};
    \end{pgfonlayer}
    \end{tikzpicture}

  \caption{\label{fig:gadget2}\textbf{Reverse gadget for $C(\theta)$.} In the ``reverse gadget'', the input state is now a FLO state, and magic is injected by the projection onto the non-FLO $M_{-\theta}$ state at the end. Recall that $G(H,H)$ is defined in Eq.~\eqref{eq:matchgate} and $H$ here denotes the Hadamard matrix (not to be confused with the Hamiltonian).}
\end{figure}


\subsection{Probability estimation for measurements on all qubits}
\label{subsec:measuring-all-qubits}

The main result of this subsection is given by the algorithm captured by the following theorem, the proof of which comprises the remainder of this section.

\begin{thm}[Probability estimation algorithm 1]\label{thm:runtime-1}
    Consider an $n$-qubit quantum circuit of the form
    \begin{align}
        U &= V_{k+1} \prod_{j=1}^k  C(\theta_j)_{n_j,n_j+1} V_j,
    \end{align}
    where the $V_j$ are arbitrary FLO unitaries and each $C(\theta_j)$ is a controlled phase gate applied to two adjacent qubits $n_j$ and $n_{j+1}$. Let $\ket{a}$, $\ket{b}$ be two, even parity, computational-basis vectors, and let
    \begin{align}
        p = \abs{\bra{b} U \ket{a}}^2,
    \end{align}
    be the associated Born-rule probability. \changed{Finally define
    \begin{align}
        \xi^* &= \prod_{j=1}^k\ext{\ket{M_{\theta_j}}}\\
            &= \prod_{j=1}^k\left(\cos\left(\frac{\theta_j}{4}\right) + \sin\left(\frac{\theta_j}{4}\right)\right)^2.
    \end{align}
    }Then, given parameters $\epsilon,\delta> 0$, there exists an algorithm that outputs an estimate $p^*$ of $p$ such that
    \begin{align}
        \mathrm{Pr}\left(\abs{p^* - p} > \epsilon \right) < \delta,
    \end{align}
    with runtime
    \begin{align}
        t = \order{\changed{(k+n)^3} \frac{\xi^*}{(\sqrt{p+\epsilon} - \sqrt{p})^2}\log\left(\frac{1}{\delta}\right)},
        \end{align}
        or, in the natural case where $\epsilon$ is small compared to $p$,
        \begin{align}
           t = \order{\changed{(k+n)^3} \xi^* \frac{p}{\epsilon^2}\log\left(\frac{1}{\delta}\right)}.
    \end{align}
\end{thm}

\begin{proof}
First, express the probability $p$ as
\begin{align}
    p &= \abs{\bra{b} V_{k+1} \prod_{j=1}^k C(\theta_j)_{n_j n_{j}+1}V_j  \ket{a}}^2\label{eqn:initial-prob-as-braket}.
\end{align}
Applying the gadgetization from Fig.~\ref{fig:gadget2} gives the following equation, valid for all qubit states $\ket{i}$, $\ket{j}$, $\ket{\phi}$, $\ket{\psi}$,
\begin{align}
   \bra{\phi}\bra{\psi} C(\theta) \ket{i}\ket{j} &= 4 \bra{\phi}\bra{M_{-\theta}}\bra{\psi}\left(G(H,H)\otimes I^{\otimes 2} \otimes G(H,H)\right)\ket{0}\ket{0}\ket{i}\ket{j}\ket{0}\ket{0}.
\end{align}
We therefore express the probability we seek to estimate as 
\begin{align}
    p &= 16^k \abs{\bra{0}^n\bra{M_{-\theta_1}}\cdots \bra{M_{-\theta_k}} \prod_{j=1}^m \tilde{V}_j \ket{0}^{\otimes(n+4k)}}^2,\label{eqn:initial-prob-as-braket2}
\end{align}
where the $\tilde{V}_j$ are FLO unitaries expressing the gadgetized circuit, and $m = \operatorname{poly}(n,k)$. Note that although swap gates are \emph{not} FLO operations, we can swap specific states, including computational basis states and the magic states $\ket{M_\theta}$ through the circuit, so we can assume the magic states are the last $4k$ qubits~\cite{Yoganathan-2019-magic-for-matchgates}. We have also added FLO operations at the start and end of the circuit to replace the input and output states $a$ and $b$ with vacuum states. Using the methods developed in Sec.~\ref{sec:phase-sensitive-flo-sim}, we can now apply the $\tilde{V}_j$ circuit to the initial vacuum state $\ket{0}^{n+4k}$ in polynomial time,  to obtain the state 
\begin{align}
    \ket{\psi} &= \prod_{j=1}^m \tilde{V}_j \ket{0}^{\otimes(n+4k)}.
\end{align}


Next, using the decomposition of Theorem~\ref{thm:decomposition-of-magic-states}, we now expand
\begin{align}
    \bigotimes_{j=1}^k\ket{M_{-\theta_j}} &= \bigotimes_{j=1}^k\left(\vphantom{A_j^j}\cos(-\theta_j/4)\ket{A(-\theta_j)}  + i \sin(-\theta_j/4)\ket{B(-\theta_j)}\right)\\
    &= \sum_{y\in\{0,1\}^k} i^{\abs{y}}\left(\prod_j \cos(-\theta_j/4)^{1-y_j} \sin(-\theta_j/4)^{y_j}\right)\ket{\tilde{y}, -\vec{\theta}},
\end{align}
where the $4k$-qubit state $\ket{\tilde{y}, \theta}$ is a tensor product of $k$ $4$-qubit states and the $j^\text{th}$ $4$-qubit state in $\ket{\tilde{y}, \theta}$ is $\ket{A(\theta_j)}$ if $y_j = 0$ and $\ket{B(\theta_j)}$ if $y_j = 1$. Substituting this into the expression for the probability, Eq.~\eqref{eqn:initial-prob-as-braket2}, we have
\begin{align}
    p &= \abs{\sum_{y\in\{0,1\}^k}4^{k} i^{-\abs{y}} \left(\prod_j \cos(-\theta_j/4)^{1-y_j} \sin(-\theta_j/4)^{y_j}\right)\bra{0}^n \bra{\tilde{y}, -\vec{\theta}} \prod_{j=1}^m \tilde{V}_j \ket{0}^{\otimes(n+4k)}}^2.
\end{align}
Now, let
\begin{align}
    \xi^* &= \left(\sum_y \abs{\left(\prod_j \cos(-\theta_j/4)^{1-y_j} \sin(-\theta_j/4)^{y_j}\right)}\right)^2\\
    &\geq \ext{\prod_j \ket{M_{-\theta_j}}},
\end{align}
where $\extno$ is the FLO-extent from Definition~\ref{def:extent}, and the upper bound holds by definition, so that the map
\begin{align}
    P:y\mapsto \frac{1}{\sqrt{\xi^*}}\abs{\left(\prod_j \cos(-\theta_j/4)^{1-y_j} \sin(-\theta_j/4)^{y_j}\right)},
\end{align}
is a probability distribution. \changed{Note that $P$ is a product probability distribution and so is easy to sample from by separately sampling from $\{0,1\}$ for each bit of $y$}. Defining the complex numbers
\begin{align}
    \alpha_y &= 4^k i^{-\abs{y}}\bra{0}^n \bra{\tilde{y}, -\vec{\theta}} \prod_{j=1}^m \tilde{V}_j \ket{0}^{\otimes(n+4k)}\operatorname{sign}\left(\prod_j \cos(-\theta_j/4)^{1-y_j} \sin(-\theta_j/4)^{y_j}\right),
\end{align}
each of which can be computed in polynomial time using the methods described in Sec.~\ref{sec:phase-sensitive-flo-sim}, we can now re-express the probability $p$ as
\begin{align}
    p &= \xi^*\abs{\sum_{y\in\{0,1\}^k}  P(y)\alpha_y}^2.
\end{align}
At this point, we need to bound the absolute value of each of the complex numbers in this sum. This bound is provided by Lemma~\ref{lem:bounding-the-norm2}, the proof of which we leave to Appendix~\ref{lem:app-bounding-norm-alpha-y}.
\begin{lem}[Bounding the norm of $\alpha_y$]\label{lem:bounding-the-norm2}
  For all $y\in\{0,1\}^k$ one has $\abs{\alpha_y} \leq 1$.
\end{lem}

We now seek to use a sampling algorithm to approximate $p$. For that, we will employ Lemma~7 of Ref.~\cite{PRXQuantum.3.020361} with the special case that here $d=1$ (complex numbers are one dimensional complex vectors). Let 
\begin{align}
    \mu = \sum_{y\in\{0,1\}^k}  P(y) \sqrt{\xi^*}\alpha_y,
\end{align}
and let $(y_j)_{j=1}^s$ be a list of $s$ samples independently sampled from the probability distribution $P$. Consider also the sample mean
\begin{align}
    m &= \frac{1}{s}\sum_j\sqrt{\xi^*} \alpha_{y_j}.
\end{align}
Then, for all $\epsilon > 0$, Lemma~7 of Ref.~\cite{PRXQuantum.3.020361} tells us that
\begin{align}
    \operatorname{Pr}(\abs{|\mu|^2 - |m|^2} \geq \epsilon) \leq 2e^2\exp\left(\frac{-s (\sqrt{p+\epsilon} - \sqrt{p})^2}{2\left(\sqrt{\xi^*} + \sqrt{p}\right)^2}\right).
\end{align}
Now, if we demand that
\begin{align}
    \operatorname{Pr}(\abs{|\mu|^2 - |m|^2} \geq \epsilon) < \delta,
\end{align}
for some $\delta > 0$, then we require $s$ to be such that
\begin{align}
  2e^2\exp\left(\frac{-s (\sqrt{p+\epsilon} - \sqrt{p})^2}{2\left(\sqrt{\xi^*} + \sqrt{p}\right)^2}\right) &< \delta.
\end{align}
Therefore,
\begin{align}
    s  &> 2\frac{\left(\sqrt{\xi^*} + \sqrt{p}\right)^2}{(\sqrt{p+\epsilon} - \sqrt{p})^2}\log\left(\frac{2e^2}{\delta}\right)
    .\label{eqn:how-many-samples-for-algorithm-1}
\end{align}

For each sample we require the application of \changed{a single FLO operation, with runtime $\order{(k+n)^3}$}, followed by the computation of inner product between the state and the computational basis vector $\ket{0}$, which also requires time $\order{(k+n)^3}$. We therefore obtain the statement of Theorem~\ref{thm:runtime-1}.
\end{proof}

Note that the runtime in Theorem~\ref{thm:runtime-1} depends on the unknown probability $p$ which we are attempting to estimate. This may be addressed in two ways. One may note that the denominator $(\sqrt{p+\epsilon} - \sqrt{p})^2$ appearing in the runtime is monotonic in $p$, and so one can take the most pessimistic choice $p = 1$. Alternatively, a significant improvement in runtime may be obtained by employing \textsc{Estimate} algorithm of Ref.~\cite{PRXQuantum.3.020361} to bound the required number of samples without knowing~$p$.


\subsection{Probability estimation for measurements of a subset of the qubits}
\label{subsec:measuring-some-qubits}

In a quantum computer, it is not generally necessary to measure every qubit. Thus, it is natural to consider scenarios where some qubits are measured and others are not. Indeed, the standard complete problem for promise-BQP, is formulated as the application of a quantum circuit followed by a single qubit computational basis-measurement~\cite{Yao-1993,Janzing-2007}. Due to Lemma~1 of Ref.~\cite{Yoganathan-2019-magic-for-matchgates}, computational basis states can be swapped through a FLO circuit using only free operations (FLO unitaries). Therefore, since we are restricting our attention to computational-basis measurements, we can assume without loss of generality that the first $w\leq n$ qubits are measured. 

The methods we will employ are very similar to those presented in the previous subsection. However, as we are not measuring every qubit we do not obtain the Born-rule probability as a simple sum of complex numbers, instead we will express it as the norm of a vector expressed as a sum of fermionic-Gaussian statevectors. An additional step will then be required to compute this norm in an efficient manner. 

\begin{thm}[Probability estimation algorithm 2]\label{thm:runtime-2}
    Consider an $n$-qubit quantum circuit of the form
    \begin{align}
        U &= V_{k+1} \prod_{j=1}^k  C(\theta_j)_{n_j,n_j+1} V_j,
    \end{align}
    where the $V_j$ are arbitrary FLO unitaries and each $C(\theta_j)$ is a controlled phase gate applied to two adjacent qubits $n_j$ and $n_{j+1}$. Let $\ket{a}$ be an even parity computational basis state on $n$ qubits, $\ket{b}$ be an even parity computational basis state on $w < n$ qubits, and
    \begin{align}
        p = \abs{\ketbra{b}{b} \otimes I^{\otimes (n-w)} U \ket{a}}^2,
    \end{align}
    be the associated Born-rule probability. \changed{Finally define
    \begin{align}
        \xi^* &= \prod_{j=1}^k\ext{\ket{M_{\theta_j}}}\\
            &= \prod_{j=1}^k\left(\cos\left(\frac{\theta_j}{4}\right) + \sin\left(\frac{\theta_j}{4}\right)\right)^2.
    \end{align}
    }Then, given parameters $\epsilon,\delta> 0$, there exists an algorithm that outputs an estimate $p^*$ of $p$ such that
    \begin{align}
        \mathrm{Pr}\left(\abs{p^* - p} > \epsilon \right) < \delta,
    \end{align}
    with runtime
    \begin{align}
        \tau &= \order{\frac{\left(\sqrt{\xi^*} + \sqrt{p}\right)^2}{\left(\sqrt{p + \frac{\epsilon}{2}} - \sqrt{p}\right)^2} \changed{\left(\frac{p}{\epsilon} + \frac{1}{2}\right)^2}\log\left(\delta\right)^2 \changed{\sqrt{n}}  (n+4k)^3}
    \end{align}
    Similarly to Theorem~\ref{thm:runtime-1}, we can simplify this expression in the case that $\epsilon$ is small compared to $p$:
    \begin{align}
        \tau &= \order{\frac{\xi^* \changed{p^3}}{\epsilon^4}\log\left(\delta\right)^2 \changed{\sqrt{n}} (n+4k)^3}.
    \end{align}
\end{thm}
\begin{proof}
We follow similar reasoning to Sec.~\ref{subsec:measuring-all-qubits} (or Appendix~C of Ref.~\cite{PRXQuantum.3.020361}), to obtain a vector $\ket{\Psi}$ of the form
\begin{align}
    \ket{\Psi} &= \frac{1}{s}\sum_{j=1}^{s} \ket{\psi_j},
\end{align}
where each $\ket{\psi_j}$ is a fermionic Gaussian state. In the above, we choose
\begin{align}
    s \geq \frac{2\left(\sqrt{\xi^*} + \sqrt{p}\right)^2}{\left(\sqrt{p + \frac{\epsilon}{2}} - \sqrt{p}\right)^2}\log\left(\frac{4 e}{\delta}\right),
\end{align}
so that, due to the vector-valued Hoeffding's inequality, given in Lemma~7 of Ref.~\cite{PRXQuantum.3.020361} 
\begin{align}
    \mathrm{Pr}\left(\abs{p - \norm{\Psi}^2_2} > \frac{\epsilon}{2}\right) < \frac{\delta}{2}.\label{eqn:vector-sampling-fail-prob}
\end{align}
We now restate a result from Appendix~B of Ref.~\cite{bravyi-gosset-17-impurity} which leads to an efficient norm-estimation algorithm for vectors expressed as sums of fermionic Gaussian statevectors.
\begin{lem}
    If $\ket{\psi}$ is any $n$-qubit statevector, $U$ is a FLO unitary chosen such that $\phi(U)$ is a random permutation matrix, $x$ is a random binary vector, and $\ket{\varphi} = U \ket{x}$, then the random variable 
    \begin{align}
        X &= 2^n \abs{\braket{\varphi}{\psi}}^2,
    \end{align}
    has mean 
    \begin{align}
        \mathbb{E}(X) &= \norm{\psi}^2_2,
    \end{align}
    and second moment
    \begin{align}
        \mathbb{E}(X^2) &\leq 2\sqrt{n} \norm{\psi}^4_2.
    \end{align}
\end{lem}
If we apply this to $\Psi$ and take $l$ samples of $X$, the sample mean has variance
\begin{align}
    \sigma_\mu^2 &= \frac{2\sqrt{n}}{l}\norm{\Psi}^4_2 \leq \frac{2\sqrt{n}}{l}\left(p+\frac{\epsilon}{2}\right)^2, \label{eqn:variance-sample-mean}
\end{align}
where the upper bound holds with probability greater than $1-\frac{\delta}{2}$ because of our choice of $\ket{\Psi}$. Rather than using the sample mean directly, we get a significant improvement in the number of samples required by employing a median-of-means type argument. We now apply Lemma~7.1 of Ref.~\cite{Devroye-2016}, which we reproduce here for convenience.
\begin{lem}\label{lem:restated-median-of-means}
    Let $Y^L = (Y_1, Y_2\hdots Y_L) \in \mathbb{R}^L$ be random variables with the same mean $\mu$ and variance bounded by~$\sigma^2$. Let $\eta > 1$ be given and let $M$ be the median of the $Y_j$. Then,
    \begin{align}
        \mathrm{Pr}\left(\abs{\mu - M} > 2 \eta \sigma \right) < \eta^{-L}. \label{eqn:median-of-means-from-reference}
    \end{align}
\end{lem}
For our purposes, we apply Lemma~\ref{lem:restated-median-of-means} with $\eta = e$, and choose $l = 32\sqrt{n}e^2 \left(p+\frac{\epsilon}{2}\right)^2 \epsilon^{-2}$, so that Eq.~\eqref{eqn:variance-sample-mean} becomes $\sigma_\mu \leq \epsilon (2e)^{-1}$. With these choices, Eq.~\eqref{eqn:median-of-means-from-reference} becomes
\begin{align}
        \mathrm{Pr}\left(\abs{\mu - M} > \frac{\epsilon}{2} \right) < e^{-L}
        \label{eqn:median-of-means-adapted}.
\end{align}
In order to obtain an estimate $\tilde{\mu}$ of $\norm{\Psi}^2_2$ satisfying
\begin{align}
    \mathrm{Pr}(\abs{\tilde{\mu} - \norm{\Psi}^2_2} > \frac{\epsilon}{2}) < \frac{\delta}{2}, \label{eqn:norm-est-fail-prob}
\end{align}
we therefore require a total of $lL$ samples satisfying
\begin{align}
    lL = 32e^2\sqrt{n} \left(\frac{p}{\epsilon} + \frac{1}{2}\right)^2\log\left(\frac{2}{\delta}\right).
\end{align}
Combining equations~\eqref{eqn:vector-sampling-fail-prob} and~\eqref{eqn:norm-est-fail-prob}, we have 
\begin{align}
    \mathrm{Pr}\left(\abs{\tilde{\mu} - p} > \epsilon \right) &\leq \mathrm{Pr}\left(\abs{\tilde{\mu} - \norm{\Psi}_2^2} +  \abs{\norm{\Psi}_2^2 - p} > \epsilon \right) \\ 
    & \leq \mathrm{Pr}\left(\abs{\tilde{\mu} - \norm{\Psi}_2^2} > \frac{\epsilon}{2} \text{ or } \abs{\norm{\Psi}_2^2 - p} > \frac{\epsilon}{2} \right)\\
    & \leq \mathrm{Pr}\left(\abs{\tilde{\mu} - \norm{\Psi}_2^2} > \frac{\epsilon}{2} \right) +\mathrm{Pr}\left(\abs{\norm{\Psi}_2^2 - p} > \frac{\epsilon}{2} \right)\\
    &< \delta,
\end{align}
where the second upper bound is just the union bound. Obtaining this performance requires computing $s lL $ inner products, each of which can be computed using the methods developed in Sec.~\ref{sec:phase-sensitive-flo-sim} in time 
\begin{align}
    \tau_{\text{sample}} &= \changed{\order{(n+4k)^3}}.
\end{align}\
Therefore, we obtain  the statement of Theorem~\ref{thm:runtime-2}.
\end{proof}


\section{Discussion of performance}
\label{sec:discussion}

Let us now provide a brief discussion of the performance of our algorithm in comparison to the prior methods. We start by comparing
Eq.~\eqref{eqn:how-many-samples-for-algorithm-1} (that captures the runtime complexity of our algorithm) to the corresponding expression of Ref.~\cite{PhysRevResearch.4.043100},
\begin{align}
    N \geq 2 W(\mathcal{E})^2 \frac{1}{\epsilon^2}\log\left(\frac{1}{\delta}\right),
\end{align}
with $W(\mathcal{E})$ denoting the fermionic non-linearity. Analytic forms for the fermionic non-linearity bounds arising from the decompositions of Ref.~\cite{PhysRevResearch.4.043100} are not given, however the fermionic nonlinearlity of a swap gate $W(\operatorname{SWAP})$ is approximately $3$ according to Fig.~1 of the same reference. Comparing their methods to Ref.~\cite{Mitarai_2021}, one can verify that in fact 
\begin{align}
    W(\mathcal{E}_\text{rot}) &= 1 + 4 \abs{\cos(\theta)\sin(\theta)} = 1 + 2 \abs{\sin(2\theta)},
\end{align}
where
\begin{align}
    \mathcal{E}_\text{rot}(\rho) &= \exp\left(i\theta \operatorname{Z}\otimes\operatorname{Z}\right)\rho \exp\left(i\theta \operatorname{Z}\otimes\operatorname{Z}\right)^\dagger,
\end{align}
is FLO-equivelent to $\operatorname{C}(4\theta)$. For a circuit with $k$ swap gates, the algorithm of Ref.~\cite{PhysRevResearch.4.043100} therefore requires
\begin{align}
    s =  9^k\times \frac{2}{\epsilon^2}\log\left(\frac{1}{\delta}\right)
\end{align}
samples, while the one described in this paper requires
\begin{align}
    s &= 2\frac{\left(\sqrt{2^k} + \sqrt{p}\right)^2}{(\sqrt{p+\epsilon} - \sqrt{p})^2}\log\left(\frac{2e^2}{\delta}\right)\\
    &\approx 2^k \frac{2}{\changed{\epsilon^2}}\log\left(\frac{2e^2}{\delta}\right),
\end{align}
\changed{assuming that $\epsilon$ is significantly smaller than $p$}. Of course, the phase-sensitive FLO simulation routines we require are more involved (and thus slower) than the standard phase-insensitive ones used by Ref.~\cite{PhysRevResearch.4.043100}. However, as the number of non-FLO gates in the circuits grows, we expect the ratio of fermionic non-linearity to FLO extent to strongly favour our methods. A comparison of the cost of adding a single $\operatorname{C}(\theta)$ gate using our methods with those of Ref.~\cite{PhysRevResearch.4.043100} is presented in Fig.~\ref{fig:comparison}.


Another interesting work on extending matchgate/FLO simulation methods to universal circuits is the recent Ref.~\cite{Mocherla-2023}. In this work, the simulation is performed by tracking the evolution of Pauli operators through a circuit comprised of FLO unitaries and non-FLO ``universality-enabling'' gates, which are essentially the same as the $\mathcal{E}_\text{rot}$ gate of Ref.~\cite{PhysRevResearch.4.043100}. The authors of Ref.~\cite{Mocherla-2023} provide runtimes which are different across different regimes, but the exponential component in each case is essentially $\order{4^k}$, where $k$ is the number of non-matchgate gates in the circuit. There are two sources for the improvement in our runtime scaling compared to algorithm of Ref.~\cite{Mocherla-2023}. Firstly, by using statevector simulation methods over density-matrix/operator methods, the scaling of $\order{4^k}$ is reduced to $2^k$. Secondly, we obtain improved runtime for controlled-phase gates which are ``between'' the identity and the worst-case $\operatorname{C}(\pi)$, controlled-Z gate. Rather than paying the full cost of a factor of $2$ for each non-free gate, our runtime is multiplied by the extent, smoothly interpolating between $1$ and $2$.

Finally, we would like to make a note about numerical errors. In order to run our algorithm, it is necessary to perform calculations involving real numbers (and matrices) with discrete \emph{floating point} numbers. Calculations involving these numbers are susceptible to floating point errors, which can build up and invalidate calculations. The most computationally intensive part of the update rules of our phase-sensitive FLO simulator is the computation of the unknown phase in the last step of applying an elementary FLO unitary. Our algorithm has the important property that errors from this step can \emph{not} build up in the usual way, upon repeated application of the update rules, since the unknown phase is either one or minus one. The only way for a floating point error from this step to lead to an error in the result is if in a single application of the rule the phase is computed with a relative error greater than one.

\begin{figure}
    \centering
        \begin{tikzpicture}
        \pgfplotsset{
every axis legend/.append style={
at={(0,1)},
anchor=north west,
},
}            \begin{axis}[
                width=0.6\textwidth,
		        height=0.4\textwidth,
		        xlabel=$\theta$,
		        ylabel style={align=center},
		        ylabel={Runtime cost for single\\ non-FLO gate},grid=both,
		        xmin=0, xmax=3.15, ymin=0, ymax=10,
		        xtick = {0, {pi/4}, {pi/2}, {3*pi/4},{pi}},
		        xticklabels={$0$, {$\frac{\pi}{4}$}, {$\frac{\pi}{2}$}, {$\frac{3\pi}{4}$}, {$\pi$}}
		        ]
		        \addplot[color=red,thick] table [x index=1,y index=2,col sep=space] {runtime-single-gate-comparison.txt};
		        \addlegendentry{$\xi(\ket{M_\theta})$}
		        \addplot[color=blue, dashed,thick] table [x index=1,y index=4,col sep=space] {runtime-single-gate-comparison.txt};
		        \addlegendentry{$W([\operatorname{C}(\theta)])^2$}
		        \end{axis}
        \end{tikzpicture}
    \caption{\label{fig:comparison}\textbf{Runtime cost of $C(\theta)$.} A comparison of the runtime cost of adding a single non-FLO controlled-phase to a quantum circuit for our simulation method and the algorithm of Ref.~\cite{PhysRevResearch.4.043100}. The notation $W([\operatorname{C}(\theta)])$ indicates the fermionic nonlinearity defined in that paper, computed for the channel $[\operatorname{C}(\theta)] : \rho \mapsto \operatorname{C}(\theta) \rho \operatorname{C}(\theta)^\dagger$. The runtime obtained by the authors of Ref.~\cite{PhysRevResearch.4.043100} increases by a factor  $W([\operatorname{C}(\theta)])^2$ for each controlled-phase gate added while ours increases by a factor of $\xi(\ket{M_\theta})$. We note a minor difference in notation, the controlled-phase gates we consider here are equivalent to the $\mathcal{E}_\text{rot}$ gates of Ref.~\cite{PhysRevResearch.4.043100} with the phase divided by 4.}
\end{figure}
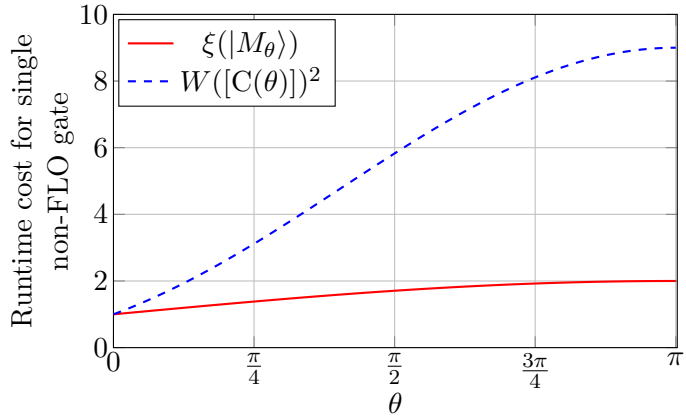


\section{Outlook}
\label{sec:outlook}

We have presented a classical algorithm for additive precision Born-rule probability estimation of universal quantum circuits composed of fermionic-linear-optical and controlled-phase gates. This way, we have extended the realm of the classically simulable quantum circuits and provided tools that can be used, e.g., for the verification of quantum devices. Moreover, Born-rule probability estimation can be used directly for estimating the expectation values of self-adjoint operators that may be expressed as a sum of polynomially many binary (e.g. $\pm 1$-valued) observables. This fact has been exploited to apply Born-rule probability estimation algorithms to the QAOA (Quantum Approximate Optimization Algorithm) algorithm in some highly specific instances~\cite{PRXQuantum.3.020361}. A natural application of the present work is the estimation of expectation-values of observables comprised of a few (polynomially many) products of Majorana operators (a natural setting in the context of quantum chemistry \cite{doi:10.1126/science.abb9811,PhysRevLett.120.110501,PhysRevApplied.9.044036,Dallaire-Demers_2019} and fermionic many-body systems \cite{Phasecraft2022,UnbiasingMonteCarlo2022}), in states which are prepared using circuits containing many FLO unitaries and few non-FLO gates. We also expect our work to find application in the increasingly important work of characterising, verifying and validating the output of quantum computers. In particular, we believe that our algorithms for the estimation of output probabilities can be directly used  for certification (via an analogue of cross-entropy benchmarking utilized in \cite{google2019supremacy}) of recently introduced quantum advantage scheme  \cite{PRXQuantum.3.020328} in which FLO circuits are initialized by a tensor product of magic states.

Another highly relevant and natural task one can consider is that of sampling from the Born-rule probability distribution or, more realistically, returning samples from a probability distribution $\epsilon$-close in total variation distance. Under plausible complexity-theoretic assumptions, \changed{this} task is strictly stronger than the Born-rule probability estimation~\cite{Pashayan2020fromestimationof}. Despite this, our methods naturally lead to a method for additive-error sampling, with the required reasoning following closely the arguments that appeared in Appendix~4.1 of Ref.~\cite{bravyi-2019} and Section~II of the Supplemental Material of Ref.~\cite{PhysRevLett.116.250501}. The essential point is that the norm-estimation procedure we sketch is, in fact, correct to \emph{multiplicative} precision, rather than additive. Thus, a procedure which estimates (conditional) probabilities for each measured qubit in turn can return samples from a probability distribution close (in $1$-norm) to the true Born-rule probability distribution. It should be emphasised that despite this, our methods do not lead to a multiplicative precision Born-rule probability estimation algorithm, since the argument from Hoeffding's inequality only provides an additive precision approximation of the statevector in question. The sampling algorithm that one obtains following this reasoning is rather impractical. The only super-polynomial component of the runtime is still given by the extent, however the necessity to estimate multiple probabilities, each to small precision, means that the polynomial factors of the runtime become significantly larger. It would be interesting, therefore, to explore dedicated sampling algorithms for the FLO+magic circuits we consider, perhaps along the lines of Ref.~\cite{PhysRevLett.128.220503} or even the Metropolis method given in Ref.~\cite{bravyi-2019}. We note, though, that the latter must be considered heuristically without understanding of the mixing time of a particular Markov chain appearing in the algorithm.

We also find it striking that the standard form we obtain for fermionic Gaussian statevectors, given in Lemma~\ref{lem:FLO-BM-form}, is, as much as is possible, identical to the CH-form of Ref.~\cite{bravyi-2019} for Clifford stabilizer states. A division of the Clifford group into essentially ``passive'' (Hadamard-free) and ``anti-passive'' (Hadamard) subgroups has recently been explored in Ref.~\cite{Bravyi-2021-hadamard-free}, resulting in a detailed exploration of the Clifford group including a decompostion that mirrors the KAK decomposition we employ here. At present, it is rather unclear why the continuous Lie-theoretic reasoning we employ, and the discrete algebraic reasoning of Ref.~\cite{Bravyi-2021-hadamard-free}, lead to what appears to be essentially the same result, albeit in different contexts. It would be interesting to explore this further, and gain understanding of whether this structure generalises to any other natural subgroups of the unitary group.



\section*{Acknowledgements}

The KK and OR-S acknowledge financial support by the Foundation for Polish Science through TEAM-NET project (contract no. POIR.04.04.00-00-17C1/18-00). O. R.-S. acknowledges funding from National Science Centre, Poland under the grant OPUS: UMO2020\\/37/B/ST2/02478. MO  acknowledges support from National Science Center, Poland within the QuantERA III Programme (No 2023/05/Y/ST2/00140 acronym Tuquan) that has received funding from the European Union’s Horizon 2020 program.

\bibliographystyle{quantum}
\bibliography{notes.bib}

\appendix


\section{Lie structure of FLO unitaries}
\label{app:lie}

In this section, we will prove that the group $\operatorname{FLO}(n)$ of fermionic linear optical unitary matrices on $n$ qubits is isomorphic to the spin group $\operatorname{Spin}(2n)$ and, in particular, that it is compact and simply connected. This is certainly not novel, and appears to be well-known in the literature (see, e.g. Refs.~\cite{bravyi-kitaev-2000,PhysRevA.63.054302}), but we were not able to find a reference with an explicit derivation, so we include one for completeness. Let $\mathfrak{flo}(n)$ be the set of $2^n\times 2^n$ complex matrices of the form
\begin{align}
  A = \frac{1}{4}\sum_{j,k=0}^{2n-1}\alpha_{jk}c_j c_k,\label{eqn:app-flo-hamiltonian-form}
\end{align}
for some real, $2n\times 2n$ antisymmetric matrix $\alpha$ (so that $\alpha$ is an element of a special orthogonal Lie algebra $\mathfrak{so}(2n)$). The elements of $\mathfrak{flo}(n)$ form a Lie algebra with the bracket being the usual commutator:
\begin{align}
  \left[\frac{1}{4}\sum_{j,k=0}^{2n-1}\alpha_{jk}c_j c_k,\frac{1}{4}\sum_{l,m=0}^{2n-1}\beta_{lm}c_l c_m \right] &= \frac{1}{4} \sum_{j,k=0}^{2n-1} [\alpha,\beta]_{jk} c_j c_k,
\end{align}
which is the reason for the factor of $1/4$ included in conventional form of FLO Hamiltonians. The elements of the Lie group $\operatorname{FLO}(n)$ are given by the exponentials of the elements of $\mathfrak{flo}(n)$. The map \mbox{$\phi:\operatorname{FLO}(n)\to \operatorname{SO}(2n)$} defined by
\begin{align}
  \phi\left(\exp\left(-\frac{1}{4}\sum_{j,k=0}^{2n-1}\alpha_{jk}c_jc_k\right)\right) := \exp(\alpha),
\end{align}
is easily seen to be a continuous Lie group \changed{anti}-homomorphism, due to the following formula of Ref.~\cite{terhal2002Classical}:
\begin{align}
 \forall U\in \operatorname{FLO}(n):\quad U c_j U^\dagger &= \sum_{j=0}^{2n-1} \phi(U)_{jk} c_k.
\end{align}
The Lie group of FLO unitaries is then closed, as it is the inverse image of the closed set $\operatorname{SO}(2n)$ under a continuous map. Since it consists of unitary matrices, it is also bounded and hence compact. 

\begin{lem}\label{lem:app-ker-of-phi}
  The kernel of $\phi$ is ${\pm I}$.
\end{lem}
\begin{proof}
  First, recall from, e.g. Ref.~\cite{gantmacher-vol1}, that for any $2n$ dimensional, real antisymmetric matrix $\alpha$, there exists a real special orthogonal matrix $R$ such that
  \begin{align}
    R \alpha R^T = \bigoplus_{j=0}^{n-1} \begin{pmatrix} 0 & \lambda_j \\ -\lambda_j & 0\end{pmatrix} =:\Lambda,
  \end{align}
  for some real $\lambda_j$. Next, choose an $\alpha$ such that $\exp(\alpha) = I$, and recall that the map $\phi$ is onto, so there exists a $U$ such that $\phi(U) = R$. Then, a simple calculation implies
  \begin{align}
    I = \phi\left(\exp\left(-\frac{1}{4} \sum_{j,k=0}^{2n-1} \alpha_{jk}c_j c_k\right)\right) \iff I = \phi\left(\exp\left(-\frac{1}{4} \sum_{j,k=0}^{2n-1} \Lambda_{jk}c_j c_k\right)\right).~\label{eqn:kernel-of-phi-block-diag-alpha}
  \end{align}
  In order to have $\exp(\Lambda) = I$, it is necessary that each $2\times 2$ block of $\Lambda$ is a logarithm of the identity in $\operatorname{SO}(2)$. Therefore, $\lambda_j = 2 k_j \pi$ for some $k_j\in\mathbb{Z}$. Since we can add any symmetric matrix to $\Lambda$ (leaving only one non-zero element in each block), the last equality of Eq.~\eqref{eqn:kernel-of-phi-block-diag-alpha} becomes
  \begin{align}
    I &= \phi\left(\exp\left( -\sum_{j=0}^{n-1} \pi k_j  c_{2j} c_{2j+1} \right)\right).
  \end{align}
  Since the individual two fermion products commute, this exponential is easily computed,
  \begin{align}
    \exp\left( \sum_{j=0}^{n-1} \pi k_j c_{2j}  c_{2j+1} \right) = \prod_{j=0}^{n-1} \left(\cos\left(k_j \pi\right)I -\sin\left(k_j \pi\right) c_{2j}c_{2j+1}\right) = (-1)^{\sum_{j=0}^{n-1} k_j} I.
  \end{align}
  By left and right-multiplying by $U$ and $U^\dagger$, it is then easy to verify that
  \begin{align}
    \exp\left(-\frac{1}{4} \sum_{j,k=0}^{2n-1} \Lambda_{jk}c_j c_k\right) = \pm I \implies \exp\left(-\frac{1}{4} \sum_{j,k=0}^{2n-1} \alpha_{jk}c_j c_k\right) = \pm I,
  \end{align}
  so the kernel of $\phi$ is $\{I, -I\}$.
\end{proof}

Since $\mathfrak{flo}(n)$ is isomorphic to the Lie algebra $\mathfrak{so}(2n)$, a maximal commutative subalgebra of $\mathfrak{flo}(n)$ is given by $\mathfrak{t}(n)$ consisting of matrices of the form given by Eq.~\eqref{eqn:app-flo-hamiltonian-form} with $\alpha$ being block-diagonal and consisting of $2\times 2$ blocks. Recall that due to, e.g., Proposition~11.7 of Ref.~\cite{hall-2015}, a maximal torus $\operatorname{T}$ in the group of FLO unitaries is given by the exponential of the maximal commutative subalgebra $\mathfrak{t}$. We then gave the following result.

\begin{lem}
  The group $\operatorname{FLO}(n)$ for $n\geq 2$ is simply connected.
\end{lem}
\begin{proof}
  Most of the work is already done for us by a standard result in the literature (see, e.g., Theorem~13.15 of Ref.~\cite{hall-2015}), which in our context states that every loop in the group of FLO unitaries is homotopic to the one in the maximal torus $\operatorname{T} = \exp(\mathfrak{t})$. Now, it suffices to show that the loops in $\operatorname{T}$ of the form
  \begin{align}
      \gamma(t) = \exp(2\pi t c_{2j} c_{2j+1}),
  \end{align}
  are contractible, as they generate all loops in~$\operatorname{T}$. Since the proof is identical for all $j$, we will directly show that the loop
  \begin{align}
    \gamma&: [0,1] \to \operatorname{FLO}(n),\\
    \gamma&:t \mapsto \exp\left(2\pi t c_0 c_1\right),\label{eqn:app-loop-contractible}
  \end{align}
  is contractible, from which the result follows. Consider the Lie subalgebra generated by the operators $c_0 c_1$, $c_1c_2$ and $c_0 c_2$, so that it consists of operators of the form
  \begin{align}
    \frac{1}{4} \sum_{j,k=0}^{2n-1} \alpha_{jk} c_j c_k = 2(r_1 c_0 c_1 + r_2 c_1 c_2 + r_3 c_0 c_2),
  \end{align}
  for some real coefficients $r_j$. By block-diagonalising $\alpha$, using the methods of Refs.~\cite{youla-1960-normal} and~\cite{hua-automorphic-1}, and directly computing the exponential, one verifies that the elements of the Lie group generated by elements of this subalgebra have the form
  \begin{align}
    U = \cos(\pi t) I + \sin(\pi t)(r_1 c_0 c_1 + r_2 c_1 c_2 + r_3 c_0 c_2),
  \end{align}
  where $r_1^2 + r_2^2 + r_3^2 = 1$.  From standard trigonometric identities, it follows that
  \begin{align}
    U = s_0 I + s_1 c_0 c_1 + s_2 c_1 c_2 + s_3 c_0 c_2,
  \end{align}
  where $\lvert\lvert{s}\rvert\rvert_2^2 := s_0^2 + s_1^2 + s_2^2 + s_3^2 = 1$. Since the loop from Eq.~\eqref{eqn:app-loop-contractible} is now embedded on the surface of the 3-sphere, it is obviously contractible.
\end{proof}

From the above result, we obtain a useful and well-known corollary: as the group $\operatorname{FLO}(n)$ is simply connected and a (double) cover of the group $\operatorname{SO}(2n)$, it is also the \emph{universal} cover of $\operatorname{SO}(2n)$, and therefore it is isomorphic to the spin group~$\operatorname{Spin}(2n)$.


\section{Proofs of the properties of passive and anti-passive FLO unitaries}
\label{app:properties}

\begin{proof}[Proof of Lemma~\ref{lem:props-of-k-type-flo}]
  Since $U_t\in\mathcal{K}$, certainly $[\alpha,\Omega] = 0$. From this, and the fact that $\Omega$ is special-orthogonal, statement $2$ follows easily from $1$. Statement $3$ is equivalent to $N U_t N^\dagger = U_t$, which follows from Eq.~\eqref{eqn:standard-symplectic-form}, and writing $N$ in terms of Majorana operators. Statement $4$ is a consequence of $3$ since the eigenvalue $0$ subspace of $N$ is one-dimensional. Statement $5$ is proved directly from Eq.~\eqref{eqn:standard-symplectic-form} and the orthogonal decomposition of the $2\times 2$ complex matrices in terms of the Pauli operators. Finally, statement $6$ is proved using the same reasoning as $5$, followed by checking conditions imposed by the special-orthogonality of $\exp(t\alpha)$.
\end{proof}

\begin{proof}[Proof of Lemma~\ref{lem:props-of-p-type}]
  From $\{\Omega,\beta\} = 0$ we see $\Omega\beta\Omega^T = -\beta$ since $\Omega$ is orthogonal. From this one easily obtains statement $2$. Now $3$ follows from applying the orthogonal decomposition of $2\times 2$ complex matrices in terms of the Pauli operators.
\end{proof}

\begin{proof}[Proof of Lemma~\ref{lem:decomp-p-type}]
  From Lemma~\ref{lem:props-of-p-type}, $\beta$ may be written in the form
  \begin{align}
    \beta = X\otimes\sigma_x + Z\otimes\sigma_z,
  \end{align}
  where $X$ and $Z$ are real, antisymmetric matrices and $\{\sigma_x, \sigma_y, \sigma_z\}$ are the $2\times 2$ Pauli matrices. Since the complex matrix $A = Z + i X$ is antisymmetric, we can apply Hua-Youla decomposition~\cite{hua-automorphic-1,youla-1960-normal}, and block diagonalize it with a unitary matrix $U$
  \begin{align}
    U (Z + iX) U^T = \Lambda\otimes i\sigma_y = \bigoplus_{j=0}^{n-1} 
    \begin{pmatrix}
      0 &\lambda_j\\
      -\lambda_j& 0
    \end{pmatrix}=  
    \begin{pmatrix}
          0 & \lambda_0 &0 & 0 \\
          -\lambda_0 & 0 & 0 & 0 \\
          0 & 0 & 0 & \lambda_1 &\\
          0 & 0 & -\lambda_1 & 0\\
            &  &  & & \ddots & &\\
            &  &  & & &  0 & \lambda_{n-1} \\
            &  &  & & &  -\lambda_{n-1} & 0\\
    \end{pmatrix},
  \end{align}
  where $\Lambda$ is a real, diagonal matrix. Now, let $c$ and $s$ be the real and imaginary parts of the $U^*$, respectively, so $U = c - is$. Set $R = c \otimes I + s\otimes i\sigma_y$ and consider the matrix
  \begin{align}
    A &= R \beta R^T\\
      &=  (c \otimes I + s\otimes i\sigma_y)(X\otimes\sigma_x + Z\otimes\sigma_z)(c^T \otimes I - s^T\otimes i\sigma_y)\\
      &= \re{U (Z+iX) U^T}\otimes\sigma_z +\im{U (Z+iX) U^T}\otimes\sigma_x\\
      &= \Lambda\otimes i\sigma_y\otimes \sigma_z.\label{eqn:symplectic-diagonalisation-p-type}
  \end{align}
\end{proof}


\section{\texorpdfstring{Proof of Lemma~\ref{lem:extent-for-F3}}{Proof of Lemma~5}}
\label{app:extent}
\begin{proof}[Proof of Lemma~\ref{lem:extent-for-F3}]
  Let $\ket{\omega}$ be any fermionic Gaussian state such that $\braket{\omega}{\psi} \neq 0$. Then
  \begin{align}
    \ket{\psi} &= \frac{\ket{\psi}\braket{\psi}{\omega}}{\braket{\psi}{\omega}}\\
               &= \frac{1}{\braket{\psi}{\omega}} \frac{1}{N}\sum_{j = 1}^{N}  V_j \ket{\omega}.\label{eqn:special-state-explicit-decomposition}
  \end{align}
  Now, Eq.~\eqref{eqn:special-state-explicit-decomposition} is a decomposition of $\ket{\psi}$ as a superposition with one norm squared of the coefficients vector given by
  \begin{align}
    \left(\sum_{j = 1}^{N}\abs{\frac{1}{N\braket{\psi}{\omega}}}\right)^2 &= \abs{\braket{\psi}{\omega}}^{-2}.
  \end{align}
  Choosing $\omega$ to be the fermionic Gaussian state with greatest overlap with $\ket{\psi}$ gives
  \begin{align}
    \xi(\psi) \leq \frac{1}{\fid{\ket{\psi}}},
  \end{align}
  while Eq.~\eqref{eqn:extent-dual-problem} with $\ket{\omega} = \ket{\psi}$ gives
  \begin{align}
    \xi(\psi) \geq \frac{1}{\fid{\ket{\psi}}}.
  \end{align}
\end{proof}


\section{Cartan's decomposition}
\label{app:Cartans-decomp}


We will first show a $KAK$-type decomposition for the special-orthogonal group, and then use the double cover property to lift this to a decomposition of the group of FLO unitaries. We require a key result from Ref.~\cite{symplecticPencils1996}, which we restate as the following theorem.

\begin{thm}[Ref.~\cite{symplecticPencils1996} - Reduction of a general real $2n\times2n$ matrix to a condensed form]\label{thm:general-2n-condensed-form}
  For any $2n\times 2n$ real matrix $A$, there exist orthogonal $2n\times 2n$ matrices $K_1$ and $K_2$ such that
  \begin{align}
    K_1 \begin{pmatrix} 0 & I \\ -I & 0 \end{pmatrix}K_1^T  = K_2 \begin{pmatrix} 0 & I \\ -I & 0 \end{pmatrix}K_2^T  = \begin{pmatrix} 0 & I \\ -I & 0 \end{pmatrix},\label{eqn:k-type-special-orthgonal-block-symplectic}
  \end{align}
  and
  \begin{align}
    K_1 A K_2 = R = \begin{pmatrix}R_{11} & R_{12}\\ R_{21} & R_{22}\end{pmatrix} = \begin{pmatrix} \begin{tikzpicture}[baseline=(current  bounding box.south), scale=0.5]
        \node[] () at (-0.2,0) {};
        \draw (1,0) -- (1,1) -- (0,1)-- (1,0);
      \end{tikzpicture} & \begin{tikzpicture}[baseline=(current  bounding box.south), scale=0.5]
        \node[] () at (1.2,0) {};
        \draw (0,0) -- (0,1) -- (1,1)-- (1,0) -- (0,0);
      \end{tikzpicture}\\  & \begin{tikzpicture}[baseline=(current  bounding box.south), scale=0.5]
        \draw (0,0) -- (0,1) -- (1,0)-- (0,0);
        \draw (0.2,1) -- (1,0.2) ;
      \end{tikzpicture}\end{pmatrix},\label{eqn:symplectic-block-factorisation}
  \end{align}
  i.e., $R_{11}$ is upper triangular, $R_{22}$ is lower Hessenberg and $R_{21}$ is zero. Moreover, the matrices $K_1$ and $K_2$ consist of $\operatorname{O}(n)$ elementary orthogonal symplectic matrices and may be determined with $\operatorname{O}(n^3)$ floating point operations.
\end{thm}
To adapt this result for our purposes, we first define some notation. Let
\begin{align}
  e_1 & = \begin{pmatrix}1 \\ 0\end{pmatrix}, & e_2 & = \begin{pmatrix}0 \\ 1\end{pmatrix}, &  E_{jk} &= e_{j} e_{k}^T,
\end{align}
and note that the block structure in Eq.~\eqref{eqn:symplectic-block-factorisation} corresponds to a tensor product structure:
\begin{equation}
  R = \sum_{j,k=0}^{1} E_{jk}\otimes R_{jk}.
\end{equation}
Contrasting with, e.g., Eq.~\eqref{eqn:standard-symplectic-form}, what we need is the structure
\begin{align}
  R = \sum_{j,k=0}^{1}  R_{jk}\otimes E_{jk},
\end{align}
so that one can perform the reduction as a simple corollary of Theorem~\ref{thm:general-2n-condensed-form}. Note that swapping the tensor product structure also transforms the symplectic form $i\sigma_y \otimes I$ appearing in Eq.~\eqref{eqn:k-type-special-orthgonal-block-symplectic} into $I\otimes i\sigma_y = \Omega$.

\begin{cor}[Reshuffled reduction of a general real 2n × 2n matrix to a condensed form]\label{cor:reshuffled-2n-condensed-form}
  For any $2n\times 2n$ real matrix $A$, there exist orthogonal $2n\times 2n$ matrices $K_1$ and $K_2$ such that
  \begin{align}
    K_1 \Omega K_1^T  = K_2 \Omega K_2^T  = \Omega,
  \end{align}
  and
  \begin{align}
    K_1 A K_2 = R = \sum_{jk}  R_{jk}\otimes E_{jk},
  \end{align}
  where
  \begin{align}
    \begin{pmatrix}R_{11} & R_{12}\\ R_{21} & R_{22}\end{pmatrix} = \begin{pmatrix} \begin{tikzpicture}[baseline=(current  bounding box.south), scale=0.5]
        \node[] () at (-0.2,0) {};
        \draw (1,0) -- (1,1) -- (0,1)-- (1,0);
      \end{tikzpicture} & \begin{tikzpicture}[baseline=(current  bounding box.south), scale=0.5]
        \node[] () at (1.2,0) {};
        \draw (0,0) -- (0,1) -- (1,1)-- (1,0) -- (0,0);
      \end{tikzpicture}\\  & \begin{tikzpicture}[baseline=(current  bounding box.south), scale=0.5]
        \draw (0,0) -- (0,1) -- (1,0)-- (0,0);
        \draw (0.2,1) -- (1,0.2) ;
      \end{tikzpicture}\end{pmatrix}.\label{eqn:symplectic-block-factorisation-reshuffled}
  \end{align}
\end{cor}

Note that if $A$ in Corollary~\ref{cor:reshuffled-2n-condensed-form} is special orthogonal, then so is $R$, and one can easily verify that $R_{12} = 0$ and $R_{11}$ is diagonal. A KAK-type factorisation of special orthogonal matrices follows easily from this.
\begin{lem}\label{lem:kak-orthogonal-matrices}
  For any real, special orthogonal $2n\times 2n$ matrix $R$, there exist symplectic orthogonal $2n\times 2n$ matrices $K_1$ and $K_2$, and real, diagonal $\frac{n}{2}\times \frac{n}{2}$ matrix $\Lambda$ such that
  \begin{align}
    R &= K_1 \exp\left({\Lambda\otimes i\sigma_y\otimes \sigma_z}\right) K_2.\label{eqn:kak-orthogonal-matrices}
  \end{align}
\end{lem}
\begin{proof}
  First, apply Corollary~\ref{cor:reshuffled-2n-condensed-form} to obtain symplectic orthogonal matrices $U$ and $V$, diagonal orthogonal matrix $D$, and lower Hessenberg orthogonal matrix $L$ such that
  \begin{align}
    R = U  \left(D \otimes E_{11} + L\otimes E_{22}\right) V.
  \end{align}
  Note that for any orthogonal $n\times n$ matrix $Q$, the matrix $Q\otimes I$ is both symplectic and orthogonal. We update $U$ and $L$ by applying the inverse of $D\otimes I$ to obtain
  \begin{align}
    R = U  D\otimes I \left(I \otimes E_{11} + D^{-1}L\otimes E_{22}\right) V.
  \end{align}
  Next, we use the well-known \changed{corollary of the real Schur decomposition~\cite{horn2012matrix}}, that any \changed{$2n\times 2n$} special orthogonal matrix may be block-diagonalized into $2\times 2$ blocks by conjugation by another special-orthogonal matrix,
  \begin{align}
    D^{-1} L &= S \bigoplus_{j=0}^{n-1} \begin{pmatrix} \cos(\lambda_j) & \sin(\lambda_j) \\ - \sin(\lambda_j) & \cos(\lambda_j)
                              \end{pmatrix}S^T\\
             &= S B S^T.
  \end{align}
  Applying this to $R$, we obtain
  \begin{align}
    R = U  (DS\otimes I) \left(I \otimes E_{11} + B \otimes E_{22}\right)  (S^T\otimes I) V.
  \end{align}
  Now, of course
  \begin{align}
    B &= \exp(\Lambda\otimes i\sigma_y),
  \end{align}
  where
  \begin{align}
    \Lambda &= \begin{pmatrix}
                 \lambda_0 & & &\\
                           & \lambda_1 & & \\
                           & & \ddots &\\
                           & & & \lambda_{\frac{n}{2}-1}
               \end{pmatrix}.
  \end{align}
  Next, we update $R$ by applying a square root of the rotation $B$, which is easily obtained from $\Lambda$:
  \begin{align}
    R = U  (DS\otimes I) (B^{\frac{1}{2}}\otimes I)\left(B^{-\frac{1}{2}} \otimes E_{11} + B^{\frac{1}{2}} \otimes E_{22}\right)  (S^T\otimes I) V.
  \end{align}
  Setting
  \begin{align}
    K_1 &= U  (DS\otimes I) (B^{\frac{1}{2}}\otimes I),\\
    A &= \left(B^{-\frac{1}{2}} \otimes E_{11} + B^{\frac{1}{2}} \otimes E_{22}\right),\\
    K_2 &= (S^T\otimes I) V,
  \end{align}
  it is easy to verify that we have obtained a decomposition of the form shown in Eq.~\eqref{eqn:kak-orthogonal-matrices}.
\end{proof}

A KAK-type factorisation of FLO unitaries follows easily from the above result.
\begin{proof}[Proof of Lemma~\ref{lem:kak-FLO-unitaries}]
  First, apply Lemma~\ref{lem:kak-orthogonal-matrices} to obtain
  \begin{align}
    \phi(U) = k_1 a k_2,
  \end{align}
  in the special-orthogonal group. Now, let
  \begin{align}
    K_1 &= \exp\left(-\frac{1}{4} \sum_{j,k=0}^{2n-1} \log(k_1)_{jk} c_j c_k\right),\\
    A &= \exp\left(-\frac{1}{4} \sum_{j,k=0}^{2n-1} \log(a)_{jk} c_j c_k\right),\\
    K_2 &= \exp\left(-\frac{1}{4} \sum_{j,k=0}^{2n-1} \log(k_2)_{jk} c_j c_k\right).
  \end{align}
  Since the map $\phi$ is an \changed{anti-homorphism} with $\ker{\phi} = \{\pm I\}$, we have that
  \begin{align}
    U &= \pm K_2 A K_1\label{eqn:flo-kak-up-to-phase}.
  \end{align}
  However, $\pm I$ are in the group of passive FLO unitaries, $\mathcal{K}$ so we just absorb the $-I$ in $K_1$ if necessary.
\end{proof}


\section{Recovering phases}
\label{app:recovering-phases}

In this appendix we address the problem of \emph{efficiently} recovering the phase $\pm I$ appearing in Eq.~\eqref{eqn:flo-kak-up-to-phase}. In order to address questions of efficiency, we specify that all FLO operators are expressed in terms of polynomial amounts of data rather than, for example, an explicit representation of a $2^{n}\times 2^{n}$ dimensional complex matrix. The method for fixing the phase is conceptually rather simple -- one simply chooses a matrix element and checks if it is the same on both sides of Eq.~\eqref{eqn:flo-kak-up-to-phase}, or if it differs by a phase. Practically, there are some complicating factors. Of course, if the matrix element turns out to be zero, this procedure does not work. Similarly, if the matrix element is very small (in magnitude) then imprecise (floating point) arithmetic on the classical computer will render the procedure meaningless. We show how to choose and compute a matrix element which is guaranteed to \changed{have absolute value equal to $1$}. We first require an intermediate lemma, showing how to compose passive (zero-preserving) FLO unitaries in a phase-sensitive way.
\begin{lem}\label{lem:app-merge-k-type-flo-unitaries}
  Given $2n\times 2n$ dimensional antisymmetric real matrices $\beta$ and $\gamma$, such that
  \begin{align}
    [\beta, \Omega] = [\gamma, \Omega] = 0,
  \end{align}
  there is a polynomial-time algorithm to compute $\xi$ such that
  \begin{align}
    \exp\left(\frac{1}{4}\sum_{j,k=0}^{2n-1}\beta_{jk}c_j c_k\right)\exp\left(\frac{1}{4}\sum_{j,k=0}^{2n-1}\gamma_{jk}c_j c_k\right) = \exp\left(\frac{1}{4}\sum_{j,k=0}^{2n-1}\xi_{jk}c_j c_k\right).
  \end{align}
\end{lem}
\begin{proof}
  We first perform the composition in the special-unitary group to obtain
  \begin{align}
    \exp(\beta)\exp(\gamma) = R.
  \end{align}
  Then, note that
  \begin{align}
    \exp\left(\frac{1}{4}\sum_{j,k=0}^{2n-1}\xi_{jk}c_j c_k\right) = \pm \exp\left(\frac{1}{4}\sum_{j,k=0}^{2n-1}\log(R)_{jk}c_j c_k\right).
  \end{align}
  The unknown phase above may be fixed by noting that
  \begin{align}
    \bra{0}\exp\left(\frac{1}{4}\sum_{j,k=0}^{2n-1}\beta_{jk}c_j c_k\right)\ket{0}\bra{0}\exp\left(\frac{1}{4}\sum_{j,k=0}^{2n-1}\gamma_{jk}c_j c_k\right)\ket{0} = \bra{0}\exp\left(\frac{1}{4}\sum_{j,k=0}^{2n-1}\xi_{jk}c_j c_k\right)\ket{0}.
  \end{align}
  We will show how to compute the first of these terms, which is sufficient since the three have the same form. First, choose $A$ and $B$ such that $\beta = A\otimes I + B\otimes i\sigma_y$, and note that $A + A^T = B - B^T = 0$. Let $U$ be the unitary which diagonalises the complex matrix $A + iB$. Since this matrix is antihermitian, it has purely imaginary eigenvalues. Let $c$ and $s$ be the real and imaginary parts of $U$, so that $U = c + is$, and define $Q = c\otimes I + s\otimes i\sigma_y$. It is easy to verify that
  \begin{align}
    Q \beta Q^T = \re{U(A + iB) U^\dagger} \otimes I + \im{U(A + iB) U^\dagger} \otimes i\sigma_y = \Lambda \otimes i\sigma_y,
  \end{align}
  where $\Lambda$ is diagonal and real, since we chose $U$ to diagonalise $A + iB$. Now, since $U$ is unitary, Lemma~\ref{lem:props-of-k-type-flo} implies that
  \begin{align}
    \bra{0}\exp\left(\frac{1}{4}\sum_{j,k=0}^{2n-1}\beta_{jk}c_j c_k\right)\ket{0} &= \bra{0}\exp\left(\frac{1}{2}\sum_{j=0}^{n-1}\Lambda_{jj}c_{2j} c_{2j+1}\right)\ket{0} = \exp\left(\frac{i}{2}\tr{\Lambda}\right).
  \end{align}
  Finally, for the sake of efficiency, note that $\tr{\Lambda}$ may be obtained without explicitly performing any diagonalisations, by noting that $\tr(\Lambda\otimes I) = -\tr((\Lambda \otimes i\sigma_y)( I  \otimes i\sigma_y)) = -\tr(Q \beta Q^T  (I  \otimes i\sigma_y) ) = -\tr(\beta (I  \otimes i\sigma_y))$.
\end{proof}

We then have the following crucial result. Note that the restriction to an even number of qubits, is a matter of convenience only; of course, we can always add an extra unused qubit to the computation.

\begin{lem}[Recovering phases]\label{lem:app-recovering-phases}
  Given Majorana fermion operators acting on an \changed{$n$-qubit Hilbert space, for even $n$} and $2n\times 2n$ dimensional real, \changed{antisymmetric} matrix $\alpha$, there is a polynomial time algorithm to determine an $\frac{n}{2}\times \frac{n}{2}$ diagonal real matrix $\Lambda$, and two $2n\times 2n$ antisymmetric real matrices $\beta$, $\gamma$ such that \mbox{$[\beta, \Omega] = [\gamma, \Omega] = 0$} and
  \begin{align}
    \exp\left(\frac{1}{4}\sum_{jk}\alpha_{jk}c_j c_k\right) = \exp\left(\frac{1}{4}\sum_{jk}\beta_{jk}c_j c_k\right)\exp\left( \frac{1}{4}\sum_{jk} (\Lambda\otimes i\sigma_y\otimes \sigma_z)_{jk} c_j c_k \right)\exp\left(\frac{1}{4}\sum_{jk}\gamma_{jk}c_j c_k\right).
  \end{align}
\end{lem}

\begin{proof}
  First, apply Lemma~\ref{lem:kak-FLO-unitaries} to obtain $\Lambda$, $\beta$ and $\gamma$ satisfying
  \begin{align}
    \exp\left(\frac{1}{4}\sum_{j,k=0}^{2n-1}\alpha_{jk}c_j c_k\right) = \pm\exp\left(\frac{1}{4}\sum_{j,k=0}^{2n-1}\beta_{jk}c_j c_k\right) \exp\left( \frac{1}{4}\sum_{j,k=0}^{2n-1} (\Lambda\otimes i\sigma_y\otimes \sigma_z)_{jk} c_j c_k \right)\exp\left(\frac{1}{4}\sum_{j,k=0}^{2n-1}\gamma_{jk}c_j c_k\right).\label{eqn:app-flo-kak-unknown-phase}
  \end{align}
  Now, employing the results of Hua~\cite{hua-automorphic-1} and Youla~\cite{youla-1960-normal}, \changed{or the real Schur decomposition~\cite{horn2012matrix}} one can block-diagonalize $\alpha$ into $2\times 2$ blocks by a special orthogonal matrix $R$:
  \begin{align}
    \alpha = R \left[\bigoplus_{j=0}^{n-1} \begin{pmatrix}
                             0 & \mu_j \\ -\mu_j & 0
                           \end{pmatrix}\right]R^T.
  \end{align}
  Define
  \begin{align}
    V = \exp\left(\changed{-}\frac{1}{4} \sum_{j,k=0}^{2n-1}\log(R)_{jk} c_j c_k\right),
  \end{align}
  so that we have
  \begin{align}
    \exp\left(\frac{1}{4}\sum_{j,k=0}^{2n-1}\alpha_{jk}c_j c_k\right) &= V \exp\left(\frac{1}{2} \sum_{j=0}^{n-1} \mu_j c_{2j}c_{2j+1}\right) V^\dagger\\
                                                            &= V \prod_{j=0}^{n-1}\left(\cos\left(\frac{\mu_j}{2}\right) I + \sin\left(\frac{\mu_j}{2}\right) c_{2j}c_{2j+1}\right) V^\dagger.
  \end{align}
  We recover the unknown phase in Eq.~\eqref{eqn:app-flo-kak-unknown-phase} by computing the ``expectation value''\footnote{Here defined as $\langle A\rangle_{\ket{\phi}} = \bra{\phi}A\ket{\phi}$ even when $A$ is not self-adjoint.} of each side of the equation in the vector $V\ket{0}$. We begin with left hand side
  \begin{align}
    \bra{0}V^\dagger \exp\left(\frac{1}{4}\sum_{j,k=0}^{2n-1}\alpha_{jk}c_j c_k\right) V \ket{0} &= \bra{0} \exp\left(\frac{1}{4}\sum_{j,k=0}^{2n-1}\left(R^T \alpha R\right)_{jk} c_j c_k\right)\ket{0}\\
    &=\bra{0}  \exp\left(\frac{1}{2}\sum_{j=0}^{n-1} \mu_j c_{2j} c_{2j+1} \right)\ket{0}\\
    &=\bra{0}  \exp\left(\frac{i}{2}\sum_{j=0}^{n-1} \mu_j Z_j\right)\ket{0}\\
     &=\exp\left(\frac{i}{2}\sum_{j=0}^{n-1}\mu_j\right).
  \end{align}
  We now calculate the expectation value of the right hand side of Eq.~\eqref{eqn:app-flo-kak-unknown-phase} the same vector
  \begin{align}
    \pm \exp\left(\frac{i}{2}\sum_{j=0}^{n}\mu_j\right) =  \bra{0}V^\dagger &\exp\left(\frac{1}{4}\sum_{j,k=0}^{2n-1}\beta_{jk}c_j c_k\right)\exp\left( \frac{1}{4}\sum_{j,k=0}^{2n-1} (\Lambda\otimes i\sigma_y\otimes \sigma_z)_{jk} c_j c_k \right)\times\nonumber\\&\exp\left(\frac{1}{4}\sum_{j,k=0}^{2n-1}\gamma_{jk}c_j c_k\right) V  \ket{0}.
  \end{align}
  For convenience, let the three exponentials in the above equation be $K_1$, $A$ and $K_2$, respectively, then this equation becomes
  \begin{align}
    \pm \exp\left(\frac{i}{2}\sum_{j=0}^{n-1}\mu_j\right) &=  \bra{0} V^\dagger K_1 A K_2 V \ket{0}\\
                                                         &=  \tr\left((K_1^\dagger V \ketbra{0}{0}  V^\dagger  K_1)A  ( K_2 K_1) \right).\label{eqn:trace-inner-product-3-terms-1}
  \end{align} 
  We now introduce the fermionic parity operator $P = (-i)^{n}\prod_{j=0}^{2n-1} c_j$. Recalling that $P$ is both unitary and self-adjoint we have
  \begin{align}
    \pm \exp\left(\frac{i}{2}\sum_{j=0}^{n-1}\mu_j\right) &=  \tr\left(P (P K_1^\dagger V \ketbra{0}{0}  V^\dagger  K_1)A  ( K_2 K_1) \right).\label{eqn:trace-inner-product-3-terms-2}
  \end{align}
    The trace may be computed, in polynomial time, using the methods detailed in Appendix~A of Ref.~\cite{bravyi-gosset-17-impurity}. We first compute generating functions for the three bracketed terms in Eq.~\eqref{eqn:trace-inner-product-3-terms-2}, and use formulas relating Gaussian integrals over Grassmann variables to Pfaffians to compute the trace as the Pfaffian of a $\operatorname{poly}(n)$ sized matrix. \changed{The formalism of Grassmann variables we use is exactly that employed in Ref.~\cite{bravyi-gosset-17-impurity}, in short a set of Grassmann variables may be thought of as a list of formal variables satisfying the anti-commutation relation $\theta_j \theta_k + \theta_k \theta_j = 0$, note that this includes the case $j=k$ implying that $\theta_j^2 = 0$ for all $j$. This relation is identical to that satisfied by the Majorana fermion operators, except when $j=k$. Quantum information theorists unfamiliar with this formalism may prefer to define each Grassmann variable $\theta_j$ to be some matrix of appropriate dimension, chosen such that the anticommutation relations are satisfied. Then the integrals over Grassmann variables we perform in the sequel may be expressed as traces. The generating functions we employ may be thought of as an isometry of (inner-product) vector spaces, mapping the algebra generated by the Majorana fermion operators to the algebra generated by the Grassmann variables. Although this isometry preserves inner products, it emphatically does not preserve algebra products, but only respects the vector-space structure of the algebras.}

  For the first bracketed term, first recall the Grassmann generating function for the projector onto the vacuum state
  \begin{align}
    \omega(\ketbra{0}{0}, \theta) &= \frac{1}{2^{2}} \exp\left(-\frac{i}{2} \theta^T M \theta \right)\\
    M  &= \bigoplus_{j=0}^n \begin{pmatrix}0 & 1 \\ -1 & 0\end{pmatrix},
  \end{align}
  where $M$ is the covariance matrix, and $\theta$ is a column vector of Grassman variables. To instead find the generating function $\omega(K_1^\dagger V \ketbra{0}{0}  V^\dagger  K_1, \theta)$ one simply replaces $M$ in the above formula by $\tilde{M} = R M R^T$ for a special orthogonal matrix $R$ which may be computed by simple matrix multiplication. We do not need to compute the product to find $R$ in a phase-sensitive way, since it appears conjugate-transposed on both sides of the projector and any overall phases will cancel.
  
  Again, following Ref.~\cite{bravyi-gosset-17-impurity}, we have
  \begin{align}
    \omega( P K_1^\dagger V \ketbra{0}{0}  V^\dagger  K_1, \theta) &= (-i)^{n} \int D\tau \exp(\tau^T \theta) \omega(K_1^\dagger V \ketbra{0}{0}  V^\dagger  K_1, \tau)\\
    &= (-i)^{n} \frac{1}{2^{n}} \int D\tau \exp\left(\tau^T \theta -\frac{i}{2} \tau^T \tilde{M}\tau\right),
  \end{align}
  where $\tau$ is another vector of $2n$ Grassman variables. We can compute the Grassman integral to obtain
  \begin{align}
      \omega( P K_1^\dagger V \ketbra{0}{0}  V^\dagger  K_1, \theta) &=(-i)^{n} \frac{1}{2^{n}} \operatorname{Pf}\left(-i \tilde{M}\right) \exp\left(-\frac{i}{2} \theta^T \tilde{M}^{-1}\theta\right)\\
      &=\frac{1}{2^{n}} \exp\left(\frac{i}{2} \theta^T \tilde{M}\theta\right).
  \end{align}
  
  The generating function for the second bracketed term in Eq.~\eqref{eqn:trace-inner-product-3-terms-2} may be directly computed. We will show the calculation in some detail, as we expect some readers will be unfamiliar with manipulations involving Grassman variables:
  \begin{align}
    A &= \exp\left( \frac{1}{4}\sum_{j,k=0}^{2n-1} (\Lambda\otimes i\sigma_y\otimes \sigma_z)_{jk} c_j c_k \right)\\
      &= \exp\left( \frac{1}{2}\sum_{j=0}^{\frac{n}{2}-1} \Lambda_{jj} \left(c_{4j} c_{4j+2}  - c_{4j+1} c_{4j+3}\right)\right)\\
      &= \prod_{j=0}^{\frac{n}{2}-1} \left(\cos\left(\frac{\lambda_j}{2}\right)I   +  \sin\left(\frac{\lambda_j}{2}\right)c_{4j} c_{4j+2} \right)\left(\cos\left(\frac{\lambda_j}{2}\right)I  -  \sin\left(\frac{\lambda_j}{2}\right)c_{4j+1} c_{4j+3} \right)\\
    \implies \omega(A, \eta) &= \prod_{j=0}^{\frac{n}{2}-1} \left(\cos\left(\frac{\lambda_j}{2}\right)   +  \sin\left(\frac{\lambda_j}{2}\right)\eta_{4j} \eta_{4j+2} \right)\left(\cos\left(\frac{\lambda_j}{2}\right)  -  \sin\left(\frac{\lambda_j}{2}\right)\eta_{4j+1} \eta_{4j+3} \right)\\
      &= \left(\prod_{j=0}^{\frac{n}{2}-1} \cos\left(\frac{\lambda_j}{2}\right)\right)^2 \prod_{j=0}^{\frac{n}{2}-1} \left(1   +  \tan\left(\frac{\lambda_j}{2}\right)\eta_{4j} \eta_{4j+2} \right)\left(1 -  \tan\left(\frac{\lambda_j}{2}\right)\eta_{4j+1} \eta_{4j+3} \right)\\
      &=\left(\prod_{j=0}^{\frac{n}{2}-1}\cos\left(\frac{\lambda_j}{2}\right)\right)^2 \prod_{j=0}^{\frac{n}{2}-1} \exp\left(\tan\left(\frac{\lambda_j}{2}\right)\left(\eta_{4j} \eta_{4j+2} - \eta_{4j+1} \eta_{4j+3}\right)\right)\\
      &=\left(\prod_{j=0}^{\frac{n}{2}-1} \cos\left(\frac{\lambda_j}{2}\right)\right)^2\exp\left(\frac{1}{2} \eta^T T \eta \right),
  \end{align}
  where $T$ is the appropriate  anti-symmetric matrix. Finally, the generating function for the third bracketed term in Eq.~\eqref{eqn:trace-inner-product-3-terms-2} may be found by employing Lemma~\ref{lem:app-merge-k-type-flo-unitaries} to obtain a matrix $\gamma$ satisfying
  \begin{align}
     K_2 K_1 = \exp\left(\frac{1}{4}\sum_{j,k=0}^{2n-1} \gamma_{jk} c_j c_k\right).
  \end{align}
  Then, we use the method given in the proof of that lemma to block-diagonalise $\gamma$ to obtain the decomposition
  \begin{align}
    K_2 K_1 &= K_3 \exp\left(\frac{1}{2}\sum_{j=0}^{n-1} \nu_j  c_{2j} c_{2j+1}\right) K_3^\dagger.
  \end{align}
  Using similar reasoning as we used above for the second term we obtain
  \begin{align}
    \omega(K_3^\dagger K_2 K_1  K_3, \phi) &=  \left(\prod_{j=0}^{n-1} \cos\left(\frac{\nu_j}{2}\right)\right) \exp\left(\sum_{j=0}^{n-1} \tan\left(\frac{\nu_j}{2}\right) \phi_{2j}\phi_{2j+1}\right)\\
   \omega( K_1 K_2, \phi) &=  \left(\prod_{j=0}^{n-1} \cos\left(\frac{\nu_j}{2}\right)\right) \exp\left(\frac{1}{2}\phi^T W \phi\right). 
  \end{align}
  We are now in a position to apply Eq.~(183) of Ref.~\cite{bravyi-gosset-17-impurity}, which we reproduce below for convenience. For any even operators $A$, $B$ and $C$ acting on an $n$ qubit Hilbert space 
  \begin{align}
    \tr\left(P A B C \right) &= (-i)^n 2^n \int D(\theta\eta\phi) \omega(A,\theta) \omega(B,\eta) \omega(C,\phi) \exp(\theta^T \eta + \eta^T\phi + \phi^T\theta).
  \end{align}
  Substituting the three generating functions, we obtain
  \begin{align}
    \tr\left(P A B C \right) &= (-i)^{n} \left(\prod_{j=0}^{\frac{n}{2}-1} \cos\left(\frac{\lambda_j}{2}\right)^2\right) \left(\prod_{k=0}^{n-1}\cos\left(\frac{\nu_k}{2}\right)\right) \\ &\qquad\int D(\theta\eta\phi) \exp\left(\frac{i}{2} \theta^T \tilde{M} \theta  +  \frac{1}{2} \eta^T T \eta+ \frac{1}{2}\phi^T W \phi + \theta^T \eta + \eta^T\phi + \phi^T\theta\right).\label{eqn:trace-inner-product-3-terms-3}
  \end{align}
  On employing the relation
  \begin{align}
    \int D(\theta)\exp\left(\frac{1}{2} \theta^T M \theta\right) &= \operatorname{Pf}(M),
  \end{align}
  the integral above can be rewritten as the Pfaffian of the matrix\\~\\
  \begin{align}
    L =\begin{pmatrix}
        i\tilde{M} & I & -I\\-I&T&I\\I&-I&W
    \end{pmatrix}
  \end{align}
  Although this algorithm is functional, we note that the numerical stability can be improved by reabsorbing the factors of $\cos\left(\frac{\lambda_j}{2}\right)$ and $\cos\left(\frac{\nu_j}{2}\right)$ back into the matrix $L$, using the property that $\operatorname{Pf}(BAB^T) = \operatorname{det}(B) \operatorname{Pf}(A)$. An appropriate choice of the matrix $B$ in this relation has determinant equal to the cosine prefactors in Eq.~\eqref{eqn:trace-inner-product-3-terms-3} and turns each appearance of $\tan$ in the matrices $T$ and $M$ into $\sin$. 
  Recalling that
  \begin{align}
      T &= \bigoplus_{j=0}^{\frac{n}{2}-1}\begin{pmatrix} 0 & 0 & \tan\left(\frac{\lambda_j}{2}\right) & 0\\
       0 & 0 & 0 & -\tan\left(\frac{\lambda_j}{2}\right)\\
       -\tan\left(\frac{\lambda_j}{2}\right) &0&0&0\\
       0&\tan\left(\frac{\lambda_j}{2}\right)&0&0
      \end{pmatrix}\\
      W&= S \bigoplus_{j=0}^{n-1}\begin{pmatrix}0 & \tan\left(\frac{\nu_j}{2}\right)\\
      -\tan\left(\frac{\nu_j}{2}\right) & 0
      \end{pmatrix} S^T,
  \end{align}
  where $S$ is the symplectic orthogonal matrix satisfying
  \begin{align}
      K_3^\dagger c_j K_3 = \sum_k S_{jk} c_k,
  \end{align}
  we define
  \begin{align}
      C_1 &= \bigoplus_{j=0}^{\frac{n}{2}-1}\begin{pmatrix}
          \cos\left(\frac{\lambda_j}{2}\right) & 0 & 0 & 0\\
          0 & \cos\left(\frac{\lambda_j}{2}\right) & 0 & 0 \\
          0&0&1&0\\
          0&0&0&1
      \end{pmatrix}\\
      C_2 &= S\bigoplus_{j=0}^{n-1}\begin{pmatrix}
          \cos\left(\frac{\nu_j}{2}\right) & 0 \\
          0 & 1
      \end{pmatrix} S^T
  \end{align}
  and set $B = I_{2n}\oplus C_1 \oplus C_2$. Finally, recalling that $\tilde{M}^{-1} = -\tilde{M}$, the standard Aitken block-diagonalization formula~\cite{schur-complement} may be used to improve the performance of the algorithm by a constant factor - requiring the computation of the Pfaffian of a $4n\times 4n$ matrix instead of a $8n\times 8n$ one.
\end{proof}
\changed{
We now show how to apply anti-passive FLO unitaries to a state expressed in the form~\ref{def:fermionic-gaussian-classical-data}.
\begin{lem}\label{lem:phase-sensitive-anti-passive}
    Given an $n$ qubit state $\ket{\psi}$ of the form given in definition~\ref{def:fermionic-gaussian-classical-data} 
    \begin{align}
        \ket{\psi} &= \omega K \exp\left(\sum_{j=0}^{\frac{n}{2}-1}\lambda_j c_{4j} c_{4j+2}\right)\ket{0}
    \end{align}
    and 
    \begin{align}
        A &= \exp\left(\sum_{j=0}^{\frac{n}{2}-1}\mu_j c_{4j} c_{4j+2}\right),
    \end{align}
    there is an $\order{n^3}$ time algorithm to compute the classical description for the state 
    \begin{align}
        A\ket{\psi} &= \tilde{\omega} \tilde{K}  \exp\left(\sum_{j=0}^{\frac{n}{2}-1}\tilde{\lambda}_j c_{4j} c_{4j+2}\right)\ket{0}
    \end{align}
\end{lem}
\begin{proof}
    The idea of the algorithm is simple, based on lemma~\ref{lem:app-recovering-phases}. We first perform the matrix multiplication in the special orthogonal group and invert the anti-homomorphism $\phi$ in order to obtain $U$ such that
    \begin{align}
        U = (-1)^a A K\exp\left(\sum_{j=0}^{\frac{n}{2}-1}\lambda_j c_{4j} c_{4j+2}\right),
    \end{align}
    for some as yet unknown $a$, we then apply the method shown in the proof of lemma~\ref{lem:app-recovering-phases} to this $U$ to obtain a $V$ such that
    $\bra{0}V^\dagger U V\ket{0}$ is an efficiently computable complex number of absolute value $1$, passive FLO unitaries $\tilde{K}$ and $\hat{K}$, and $\frac{n}{2}$ real numbers $\tilde{\lambda}_j$ such that
    \begin{align}
        U = \tilde{K} \exp\left(\sum_{j=0}^{\frac{n}{2}-1}\tilde{\lambda}_j c_{4j} c_{4j+2}\right)\hat{K}.
    \end{align}
    It remains only to compute the complex number
    \begin{align}
        \bra{0}V^\dagger A K\exp\left(\sum_{j=0}^{\frac{n}{2}-1}\lambda_j c_{4j} c_{4j+2}\right) V\ket{0} &= (-1)^a \bra{0}V^\dagger U V\ket{0},
    \end{align}
    from which we can obtain $a$. This can be achieved by essentially the same methods as employed in the proof of lemma~\ref{lem:app-recovering-phases}. We first rearrange
    \begin{align}
        \bra{0}V^\dagger A K\exp\left(\sum_{j=0}^{\frac{n}{2}-1}\lambda_j c_{4j} c_{4j+2}\right) V\ket{0} &= \tr\left( A^\dagger V \ketbra{0}{0}V^\dagger A K \exp\left(\sum_{j=0}^{\frac{n}{2}-1}(\lambda_j+\mu_j) c_{4j} c_{4j+2}\right) \right),
    \end{align}
    and then note that this trace has essentially exactly the same form as the one in equation~\eqref{eqn:trace-inner-product-3-terms-2}, and may be computed in the same way. Absorbing a factor of $(-1)$ into $\tilde{K}$ if necessary we finally obtain
    \begin{align}
        A\ket{\psi} &= \hat{\omega} \tilde{K}  \exp\left(\sum_{j=0}^{\frac{n}{2}-1}\tilde{\lambda}_j c_{4j} c_{4j+2}\right)\hat{K}\ket{0}\\
        &= \tilde{\omega} \tilde{K}  \exp\left(\sum_{j=0}^{\frac{n}{2}-1}\tilde{\lambda}_j c_{4j} c_{4j+2}\right)\ket{0}.
    \end{align}
\end{proof}}
\changed{
We now show how to compute the inner product between two states expressed in the form given in Definition~\ref{def:fermionic-gaussian-classical-data}.
\begin{lem}
  Given two $n$-qubit statevectors expressed in the form
  \begin{align}
      \ket{\psi_1} &= \omega_1 K_1 A_1 \ket{0}\\
      \ket{\psi_2} &= \omega_2 K_2 A_2 \ket{0},
  \end{align}
  where each $\omega_j\in\mathbb{C}$, each $K_j$ is a passive FLO unitary and each $A_j$ is a commuting anti-passive FLO unitary then there is an $\order{n^3}$ algorithm to compute the inner product
  \begin{align}
      \braket{\psi_1}{\psi_2} = \omega_1^*\omega_2 \bra{0}A_1^\dagger K_1^\dagger K_2 A_2\ket{0}.
  \end{align}
\end{lem}
\begin{proof}
 Again we assume without loss of generality that $n$ is even. We first express the inner product as a trace
 \begin{align}
     \omega_1^*\omega_2 \bra{0}A_1^\dagger K_1^\dagger K_2 A_2\ket{0} &= \omega_1^*\omega_2\tr\left(\ketbra{0}{0} A_1^\dagger K_1^\dagger K_2 A_2\right)\\
     &=\omega_1^*\omega_2\tr\left(A_1\ketbra{0}{0} A_1^\dagger (K_1^\dagger K_2) (A_2A_1^\dagger)\right),
 \end{align}
  now $A_2$ and $A_1^\dagger$ may be merged in the obvious way, simply by adding the two vectors of real numbers which define them, while $K_1^\dagger$ and $K_2$ may be combined using lemma~\ref{lem:app-merge-k-type-flo-unitaries}. This trace then has essentially the same form as the one appearing in equation~\eqref{eqn:trace-inner-product-3-terms-2} in lemma~\ref{lem:app-recovering-phases} and may be computed in the same way.
\end{proof}}
Since we are considering statevectors, including phase information, it makes sense to apply projectors to obtain possibly subnormalised states.
\begin{lem}[Applying projectors]\label{lem:app-applying-projectors}
  Given a fermionic-Gaussian statevector of the form
\begin{align}
  \ket{\psi} = \omega K A \ket{0},
\end{align}
where $K\ket{0} = \ket{0}$ and $A = \exp\left({\sum_k\lambda_k c_{4k}c_{4k+2}}\right)$, and a single qubit projector $P_j = a_j a_j^\dagger$. There is an $\order{N^3}$ algorithm to obtain $\omega^\prime$, $K^\prime$ and $A^\prime$ of the same form such that
\begin{align}
    a_j a_j^\dagger \ket{\psi} = \omega^\prime K^\prime A^\prime \ket{0}.
\end{align}

\end{lem}
\begin{proof}
First, we seek an update rule of the form 
\begin{align}
  K^\dagger a_j K &= \sum_k T_{jk} a_k.
\end{align}
This is easily obtained from
\begin{align}
  K^\dagger c_j K &= \sum_k R_{jk} c_k\\
  a_j &= \frac{1}{2}\left(c_{2j} + i c_{2j+1}\right),
\end{align}
and has the form
\begin{align}
  K^\dagger a_j K &= \frac{1}{2} K^\dagger c_{2j} K  + \frac{i}{2} K^\dagger c_{2j+1} K\\
                  &=\frac{1}{2} \sum_k \left(R_{2j,j} + i R_{2j+1, k}\right)c_{k}\\
                  &=\frac{1}{2} \sum_k \left(R_{2j,2k} + i R_{2j+1, 2k}\right)c_{2k} + \left(R_{2j,2k+1} + i R_{2j+1, 2k+1}\right)c_{2k+1} \\
                  &=\frac{1}{2} \sum_k \left(R_{2j,2k} + i R_{2j+1, 2k}\right)(a_k + a_k^\dagger) + \left(R_{2j,2k+1} + i R_{2j+1, 2k+1}\right)(-i)(a_k - a_k^\dagger) \\
                  &=\frac{1}{2} \sum_k \left(R_{2j,2k} + R_{2j+1, 2k+1}+  i R_{2j+1, 2k} - i R_{2j,2k+1}\right)a_k.
\end{align}
Note here that there is no $a_j^\dagger$ component in the sum exactly because $K$ is zero-preserving. Of course
\begin{align}
  K^\dagger a_j^\dagger K &= \left(K^\dagger a_j K\right)^\dagger\\
                          &=\frac{1}{2} \sum_k \left(R_{2j,2k} + R_{2j+1, 2k+1} -i R_{2j+1, 2k} + i R_{2j,2k+1}\right)a_k^\dagger.
\end{align}
Compressing the notation slightly, we now have
\begin{align}
  a_j a_j^\dagger K A \ket{0} &= K K^\dagger a_j a_j^\dagger K A \ket{0}\\ 
                              &= K \left(\sum_kx_k a_k\right)\left(\sum_kx_k^*a_k^\dagger\right)A\ket{0},\label{eqn:projector-move-1}
\end{align}
and we need to move $\sum_jx_j a_j$ and $\sum_jx_j^*a_j^\dagger$ through the unitary $A$. Using similar reasoning to above and the explicit form
\begin{align}
  A = \prod_j\left(\cos\lambda_j I + \sin\lambda_j c_{4j}c_{4j+2}\right),
\end{align}
it is easy to find $S$ and $T$ such that
\begin{align}
  A^\dagger a_j A &= \sum_j S_{jk} a_k + T_{jk} a_k^\dagger,\label{eqn:projector-move-2} \\
  A^\dagger a_j^\dagger A &= \sum_j S_{jk}^* a_k^\dagger + T_{jk}^* a_k.\label{eqn:projector-move-3}
\end{align}
Since $A$ is not zero-preserving, we have $a_k^\dagger$ components in the expression for $A^\dagger a_k A$ and vice-versa. However, noting that $a_j\ket{0} = \changed{0}$, we obtain 
\begin{align}
  a_j a_j^\dagger K A \ket{0} &= K A \sum_k \left(s_k a_k + t_k a_k^\dagger\right)\sum_l u_l a_l^\dagger \ket{0},
\end{align}
\changed{for some complex vectors $s$, $t$ and $u$ which may be computed by combining equations~\eqref{eqn:projector-move-1}, ~\eqref{eqn:projector-move-2} and~\eqref{eqn:projector-move-3}.}
Recalling the anti-commutation relations for the creation and annihilation operators, this becomes
\begin{align}
  a_j a_j^\dagger K A \ket{0} &= K A\left( \sum_k s_k u_k I  + \sum_{kl} t_k u_l a_k^\dagger a_l^\dagger\right)\ket{0}\\
                              &= K A\left( \sum_k s_k u_k I  + \sum_{kl} t_k u_l a_k^\dagger a_l^\dagger\right)\ket{0}.
\end{align}
For any real vector $v$, the anti-commutation relations imply
\begin{align}
  \sum_{jk} v_j v_k a_j^\dagger a_k^\dagger = 0,
\end{align}
so defining
\begin{align}
  t^\prime = t - \frac{(t\cdot u)}{u\cdot u} u,
\end{align}
so that $t^\prime \cdot u = 0$, we obtain
\begin{align}
  a_j a_j^\dagger K A \ket{0} &= K A\left( \sum_k s_k u_k I  + \sum_{kl} t_k u_l a_k^\dagger a_l^\dagger\right)\ket{0}\\
                              &=K A\left( \sum_k s_k u_k I  + \sum_{kl} t_k^\prime u_l a_k^\dagger a_l^\dagger\right)\ket{0}.
\end{align}
Since $t^\prime$ and $u$ are orthogonal, we can find an $SO(n)$ matrix $r$ such that
\begin{align}
  R t^\prime &= e_0\\
  R u &= e_1,
\end{align}
where $e_j$ is the $j^\text{th}$ column of the $n\times n$ identity matrix. This provides a passive FLO unitary $K_1$ such that
\begin{align}
    a_j a_j^\dagger K A \ket{0} &=K A\left( \sum_k s_k u_k I  + \sum_{kl} t_k^\prime u_l a_k^\dagger a_l^\dagger\right)\ket{0}\\
                                &= K A K_1^\dagger\left( a I  + b a_0^\dagger a_1^\dagger\right) K_1 \ket{0}\\
                                &= \alpha K A K_1^\dagger\left( \cos{\lambda}\, I  + \sin{\lambda}\, a_0^\dagger a_1^\dagger\right) K_1 \ket{0}\\
                                &= \tilde{\alpha} K A K_1^\dagger\left( \cos{\lambda}\, I  + \sin{\lambda}\, c_0 c_2\right) \ket{0}, \label{eqn:applying-projectors-all-unitary}
\end{align}
for real constants $a$, $b$, $\lambda$, and $\alpha$, and complex $\tilde{\alpha}$. Now, each operator applied to $\ket{0}$ in Eq.~\eqref{eqn:applying-projectors-all-unitary} is a FLO unitary, and the product may be re-expressed in $KAK$ form using the methods we have described above.
\end{proof}


\section{\texorpdfstring{Bounding $\lvert\alpha_y\rvert$}{Bounding the norm of alpha}}
\label{lem:app-bounding-norm-alpha-y}
\FloatBarrier 
\begin{proof}[Proof of Lemma~\ref{lem:bounding-the-norm2}]
  Recall
  \begin{align}
    \abs{\alpha_y} &= 4^k \abs{\bra{0}^n \bra{\tilde{y}, -\vec{\varphi}} \prod_{j=1}^m \tilde{V}_j \ket{0}^{\otimes(n+4k)}}
  \end{align}
  and the initial state has norm $1$, so we need to show that each gadget we introduce reduces the norm of the state by a factor of $\frac{1}{4}$. Each gadget will have the form shown in figure~\ref{fig:gadget2}, with $\ket{M_\theta}$ replaced by either $\ket{A(\theta_j)}$ or $\ket{B(\theta_j)}$ depending on the value of $y_j$. It makes no difference to the following argument whether we choose $A$ or $B$, for notational convenience we will assume that $y_j=0$ so we have state $A(\theta_j)$. In this case the gadget has the form shown in figure~\ref{fig:gadget-right-specialized}
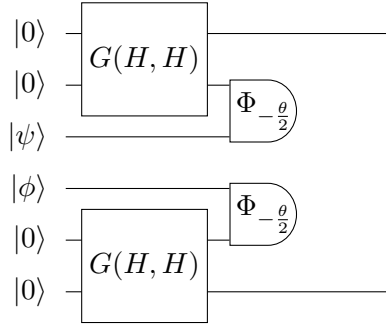
\begin{figure}
  \centering

\pgfdeclarelayer{background}
\pgfdeclarelayer{foreground}
\pgfsetlayers{background,main,foreground}
        \begin{tikzpicture}
    \tikzstyle{operator} = [draw,fill=white,minimum size=1.5em] 
    \tikzstyle{phase} = [draw,fill,shape=circle,minimum size=0.2cm,inner sep=0pt]
    \tikzstyle{multiQubitGate} = [fill=white,draw=black,shape=rectangle, minimum height=1.5cm, inner sep=1mm]
    \tikzstyle{singleQubitMeasurement} = [fill=white,draw=black,shape=rounded rectangle,rounded rectangle left arc=none,rounded rectangle arc length=130,inner sep=0.1cm, outer sep=0cm]
    %
    \matrix[row sep=0.3cm, column sep=1cm] (circuit) {
    \node[label=left:{$\ket{0}$}] (q1) {}; & 
    \node[phase] (G11) {}; & &
    \node (skip) {}; &
    \coordinate (end1); \\
    \node[label=left:{$\ket{0}$}] (q2) {}; & 
    \node[phase] (G12) {}; &  
    \node (M2) {}; &
    \coordinate (end2);\\
    \node[label=left:{$\ket{\psi}$}] (q3) {}; & & \node(M3) {};
    &
    \coordinate (end3);\\
    \node[label=left:{$\ket{\phi}$}] (q4){}; & & \node(M4) {};
    &
    \coordinate (end4);\\
    \node[label=left:{$\ket{0}$}] (q5) {}; & 
    \node[phase] (G21) {}; & 
    \node(M5) {}; &
    \coordinate (end5); \\
    \node[label=left:{$\ket{0}$}] (q6) {};  &
    \node[phase] (G22) {}; &
    \node (end7) {}; & &
    \coordinate (end6); \\
    };
     \begin{pgfonlayer}{foreground}
       \coordinate (rectbottomleft) at ($(M2)!1.1!(M3)$);
       \coordinate (recttopleft) at ($(M3)!1.1!(M2)$);
       \coordinate (recttopright) at ($(recttopleft) + (0.5cm, 0)$);
       \coordinate (rectbottomright) at ($(rectbottomleft) + (0.5cm, 0)$);
       \newdimen\mydim
       \pgfextracty{\mydim}{\pgfpointscale{0.5}{\pgfpointdiff{\pgfpointanchor{rectbottomright}{center}}{\pgfpointanchor{recttopright}{center}}}}
       \draw[fill=white, draw=black] (recttopleft) -- (rectbottomleft) -- (rectbottomright) arc[start angle=270, delta angle=180, x radius=.375cm, y radius=\mydim] -- cycle;
       \coordinate (rectmidleft) at ($(M2)!0.5!(M3)$);       
       \node (bigmeasurementlabel) at ($(rectmidleft) + (0.45cm, 0)$) {$\displaystyle{\Phi_{-\frac{\theta}{2}}}$};
    \end{pgfonlayer}
    \begin{pgfonlayer}{foreground}
       \coordinate (rectbottomleft) at ($(M4)!1.1!(M5)$);
       \coordinate (recttopleft) at ($(M5)!1.1!(M4)$);
       \coordinate (recttopright) at ($(recttopleft) + (0.5cm, 0)$);
       \coordinate (rectbottomright) at ($(rectbottomleft) + (0.5cm, 0)$);
       \newdimen\mydim
       \pgfextracty{\mydim}{\pgfpointscale{0.5}{\pgfpointdiff{\pgfpointanchor{rectbottomright}{center}}{\pgfpointanchor{recttopright}{center}}}}
       \draw[fill=white, draw=black] (recttopleft) -- (rectbottomleft) -- (rectbottomright) arc[start angle=270, delta angle=180, x radius=.375cm, y radius=\mydim] -- cycle;
       \coordinate (rectmidleft) at ($(M4)!0.5!(M5)$);       
       \node (bigmeasurementlabel) at ($(rectmidleft) + (0.45cm, 0)$) {$\displaystyle{\Phi_{-\frac{\theta}{2}}}$};
    \end{pgfonlayer}
        
        
    \begin{pgfonlayer}{background}
        \draw[] (q1) -- (end1)  (q2) -- (end2) (q3) -- (end3) (q4) -- (end4) (q5) -- (end5) (q6) -- (end6); 
        \coordinate (rectbottomleft) at ($(M2)!1.2!(M3)$);
        \coordinate (recttopleft) at ($(M3)!1.2!(M2)$);
        \coordinate (recttopright) at ($(M2)!1.2!(end2)$);
        \coordinate (rectbottomright) at ($(M3)!1.2!(end3)$);
        \draw[fill=white, draw=white] (rectbottomleft) -- (recttopleft) -- (recttopright) -- (rectbottomright) -- cycle;
        \coordinate (rectbottomleft) at ($(M4)!1.2!(M5)$);
        \coordinate (recttopleft) at ($(M5)!1.2!(M4)$);
        \coordinate (recttopright) at ($(M4)!1.2!(end4)$);
        \coordinate (rectbottomright) at ($(M5)!1.2!(end5)$);
        \draw[fill=white, draw=white] (rectbottomleft) -- (recttopleft) -- (recttopright) -- (rectbottomright) -- cycle;
    \end{pgfonlayer}
    \begin{pgfonlayer}{foreground}
        \node[multiQubitGate, at = ($(G11)!0.5!(G12)$)] (G1) {$G(H,H)$};
        \node[multiQubitGate, at = ($(G21)!0.5!(G22)$)] (G2) {$G(H,H)$};
    \end{pgfonlayer}
    \end{tikzpicture}
    \caption{The ``reverse gadget'', shown in in Fig.~\ref{fig:gadget2}, here with the magic state $\ket{M_{-\theta}}$ replaced by the state $\ket{A_{-\theta}} = \frac{1}{2}\left(\ket{00} + e^{-i\frac{\theta}{2}}\ket{11}\right)\otimes \left(\ket{00} + e^{-i\frac{\theta}{2}}\ket{11}\right)$ appearing in the decomposition \eqref{eqn:magic-state-orthogonal-decomposition}. The state $\ket{\Phi_{\theta}}$ in the projectors indicates the generalisation of the Bell-state with arbitrary phase $\ket{\Phi_{\theta}} = \frac{1}{\sqrt{2}}\left(\ket{00} + e^{i\theta}\ket{11}\right)$.}\label{fig:gadget-right-specialized}
\end{figure}

    We wish to examine what happens when we input two qubits of an arbitrary $n$ qubit state $\ket{\phi}$ on the third and fourth branches of this gadget. First recall from the definition of $G$, in equation~\eqref{eq:matchgate}, that upon applying it to two Hadamard gates one obtains
    \begin{align}
        G(H,H) &= \frac{1}{\sqrt{2}}\begin{pmatrix} 1 & 0 & 0 &1\\0&1&1&0\\0&1&-1&0\\1&0&0&-1\end{pmatrix},
    \end{align}
    so $G(H,H)\ket{00} = \frac{1}{\sqrt{2}}\left(\ket{00} + \ket{11}\right)$.
    
    Schmidt decomposing the input state we obtain
    \begin{align}
        \ket{\text{input}} = \sum_{a,b\in\{0,1\}} \frac{\lambda_{ab}}{2} \left(\ket{00} + \ket{11}\right)\ket{a}\ket{b} \left(\ket{00} + \ket{11}\right)\ket{\phi_{ab}},
    \end{align}
    where $\ket{a}$ and $\ket{b}$ are computational basis states, the $\ket{\phi_{ab}}$ a set of $4$ orthonormal vectors on the remaining $n-2$ qubits comprising the $n$ qubit state and the $\lambda_{ab}$ are positive real numbers whose squares sum to the square of the norm of the input state. Now since
    \begin{align}
        \ketbra{\Phi_\theta}{\Psi_\theta} &= \frac{1}{2} \left(\ket{00} + e^{i\theta}\ket{11}\right)\left(\bra{00} + e^{-i\theta}\bra{11}\right)\\
         \ketbra{\Phi_\theta}{\Psi_\theta} \ket{00} &= \frac{1}{\sqrt{2}}\ket{\Psi_\theta}\\
         \ketbra{\Phi_\theta}{\Psi_\theta} \ket{11} &= e^{-i\theta}\frac{1}{\sqrt{2}}\ket{\Psi_\theta}\\
         \ketbra{\Phi_\theta}{\Psi_\theta} \ket{01} &= 0\\
         \ketbra{\Phi_\theta}{\Phi_\theta} \ket{10} &= 0,
    \end{align}
    we have
    \begin{align}
        \left(I\otimes\ketbra{\Phi_{-\frac{\theta}{2}}}{\Phi_{-\frac{\theta}{2}}}\otimes \ketbra{\Phi_{-\frac{\theta}{2}}}{\Phi_{-\frac{\theta}{2}}}\otimes I\right)\ket{\text{input}} = &\frac{\lambda_{00}}{4}\ket{0}\ket{\Phi_{-\frac{\theta}{2}}}\ket{\Phi_{-\frac{\theta}{2}}}\ket{0} \ket{\phi_{00}} +\\ 
        &\frac{\lambda_{01}}{4}e^{i\frac{\theta}{2}}\ket{0}\ket{\Phi_{-\frac{\theta}{2}}}\ket{\Phi_{-\frac{\theta}{2}}}\ket{1} \ket{\phi_{01}} +\\ 
        &\frac{\lambda_{10}}{4}e^{i\frac{\theta}{2}}\ket{1}\ket{\Phi_{-\frac{\theta}{2}}}\ket{\Phi_{-\frac{\theta}{2}}}\ket{0}\ket{\phi_{10}} +\\ 
        &\frac{\lambda_{11}}{4}e^{i\theta}\ket{1}\ket{\Phi_{-\frac{\theta}{2}}}\ket{\Phi_{-\frac{\theta}{2}}}\ket{1} \ket{\phi_{11}},
    \end{align}
    evidently the norm of this vector is a quarter of the norm of the vector we started with.
\end{proof}

\section{FLO-fidelity of a tensor product of two maximally non-Gaussian states}
\label{app:fidelity-of-a8-pair}
Let
\begin{align}
    \ket{a_8} = \frac{1}{\sqrt{2}}\left(\ket{0000} + \ket{1111}\right),
\end{align}
be the state defined by Bravyi in Ref.~\cite{PhysRevA.73.042313}. We show that 
\begin{align}
    \fid{\ket{a_8} \ket{a_8} } &= \fid{\ket{a_8}}^2 = \frac{1}{4}.
\end{align}
Recall
\begin{align}
   \fid{\ket{a_8} \ket{a_8} } &= \max_{U\in\mathrm{FLO}(8)} \abs{\bra{a_8}\bra{a_8}U\ket{0} }^2.
\end{align}
We know from Ref.~\cite{BOTERO200439} that the $8$ qubit FLO state $\ket{\phi} = U\ket{0}$ may be decomposed into BCS-like entangled modes. That is there exist $4$-qubit FLO unitaries $U_A$, $U_B$ and $4$ angles $\theta_j$ such that 
\begin{align}
    \ket{\phi} &= U_A\otimes U_B \prod_j\left(\cos(\theta_j)\ket{0}_A\ket{0}_B +  \sin(\theta_j)\ket{1}_A\ket{1}_B\right),
\end{align}
so the optimization above becomes
\begin{align}
    &\fid{\ket{a_8} \ket{a_8} }\nonumber\\ &= \max_{U_A, U_B, \theta_j} \abs{\bra{a_8}\bra{a_8}U_A\otimes U_B \prod_j\left(\cos(\theta_j)\ket{0}_A\ket{0}_B +  \sin(\theta_j)\ket{1}_A\ket{1}_B\right)}^2.
\end{align}
Abbreviating with
\begin{align}
    c_j &= \cos(\theta_j)\\
    s_j &= \sin(\theta_j)\\
    t(y) &= \prod_j c_j^{1-y_j}s_j^{y_j},\label{eqn:abbreviating-t}
\end{align}
for each length $4$ bitstring, $y$ we obtain
\begin{align}
    \fid{\ket{a_8} \ket{a_8} } &= \max_{U_A, U_B, \theta_j} \abs{\sum_{y\in\{0,1\}^4} t(y)\bra{a_8}U_A \ket{y}\bra{a_8}U_B \ket{y}}^2.
\end{align}
At this point we note that we can restrict our attention to $c_j \geq 0$, $s_j \geq 0$, by appropriately defining $U_A$ and $U_B$. The case where either $c_j$ or $s_j$ is zero is uninteresting, as the state becomes a product state and may be addressed separately. Therefore we only consider $t(y) > 0$, it follows that $t$ defines an inner product on length $16$ vectors (indexed by $y$)
\begin{align}
    \langle a, b \rangle_t = \sum_y t(y) a_y^* b_y.
\end{align}
We apply the Cauchy-Schwartz inequality
\begin{align}
    |\langle a, b \rangle_t|^2 \leq \langle a, a \rangle_t \langle b, b \rangle_t
\end{align}
to obtain
\begin{align}
    \fid{\ket{a_8} \ket{a_8} } &\leq \max_{U_A, U_B, \theta_j} \abs{\sum_{y\in\{0,1\}^4} t(y)\bra{a_8}U_A^\dagger \ketbra{y}{y}U_A \ket{a_8}} \abs{\sum_{y\in\{0,1\}^4} t(y)\bra{a_8}U_B^\dagger \ketbra{y}{y}U_B \ket{a_8}}.
\end{align}
Obviously the $A$ and $B$ terms are exactly the same, so the result we seek reduces to showing 
\begin{align}
    \max_{U_A, \theta_j} \abs{\sum_{y\in\{0,1\}^4} t(y)\bra{a_8}U_A^\dagger \ketbra{y}{y}U_A \ket{a_8}} \leq \fid{\ket{a_8}}.
\end{align}
We know from Ref~\cite{Oszmaniec2014} that the orbit $\operatorname{FLO}(4)\ket{a_8}$ consists of states with purely real or purely imaginary coefficients when expressed in the basis $\ket{\eta_j}$ which appears in the proof of Theorem~\ref{thm:decomposition-of-magic-states}. Given the explicit form for the vectors in this basis we obtain
\begin{align}
    \begin{split}
    \sum_{y\in\{0,1\}^4} t(y)\bra{a_8}U_A^\dagger \ketbra{y}{y}U_A \ket{a_8} &= \frac{1}{2} \left(t(0000) + t(1111)\right) (|r_1|^2 + |r_2|^2) \\
    &\quad+ \frac{1}{2} \left(t(0011) + t(1100)\right) (|r_3|^2 + |r_4|^2) \\
    &\quad+ \frac{1}{2} \left(t(0101) + t(1010)\right) (|r_5|^2 + |r_6|^2) \\
    &\quad+ \frac{1}{2} \left(t(1001) + t(0110)\right) (|r_7|^2 + |r_8|^2),
    \end{split}\label{eqn:fidelity-two-copies-of-a8}
\end{align}
where $U_A\ket{a_8} = \sum_j r_j\ket{\eta_j}$. The expression appearing in~\eqref{eqn:fidelity-two-copies-of-a8} is easy to optimize. Since the $|r_j|^2$ sum to $1$ and are non-negative, the optimizer will have $r$ chosen such that all of the weight in the expression is applied to whichever of the $t(y) + t(\bar{y})$ terms is largest. Then we simply have to optimize an expression of the form
\begin{align}
    t(0000) + t(1111) &= \cos(\theta_1)\cos(\theta_2)\cos(\theta_3)\cos(\theta_4) + \sin(\theta_1)\sin(\theta_2)\sin(\theta_3)\sin(\theta_4).
\end{align}
By choosing each $\theta$ equal to $0$ we see that the value $1$ is achievable, while by noting 
\begin{align}
    t(0000) + t(1111) \leq \cos(\theta_1)\cos(\theta_2) + \sin(\theta_1)\sin(\theta_2) = \cos(\theta_1 - \theta_2) \leq 1
\end{align}
shows this value is the optimum. We can insert this value into the expression~\eqref{eqn:fidelity-two-copies-of-a8} to obtain
\begin{align}
    \sum_{y\in\{0,1\}^4} t(y)\bra{a_8}U_A^\dagger \ketbra{y}{y}U_A \ket{a_8} \leq \frac{1}{2},
\end{align}
or, equivalently
\begin{align}
   \fid{\ket{a_8} \ket{a_8} } &= \fid{\ket{a_8}}^2 = \frac{1}{4}.
\end{align}
This method of proof generalises quite broadly. Define the $2^n\times 2^n$ diagonal matrix
\begin{align}
    Y(\theta) &= \sum_{y\in\{0,1\}^n }t(y)\ketbra{y}{y}
\end{align}
with $t(y)$ depending on $n$ angles $\theta_j$ as in equation~\eqref{eqn:abbreviating-t}. Then say an $n$ qubit pure $\ket{\psi}$ is \emph{FLO-aligned} if
\begin{align}
    \bra{\psi} U^\dagger Y U \ket{\psi} \leq \fid{\ket{\psi}},
\end{align}
for all $\theta$ and $U\in\mathrm{FLO}(n)$. Then if $\ket{\psi_1}$ and $\ket{\psi_2}$ are \emph{FLO-aligned} then
\begin{align}
    \fid{\ket{\psi_1}\ket{\psi_2}} = \fid{\ket{\psi_1}}\fid{\ket{\psi_2}}.
\end{align}

\end{document}